\documentclass{theoretics}
\usepackage[T1]{fontenc}

\usepackage{amsfonts,amsthm,amsmath,amssymb}
\usepackage[capitalise,noabbrev,nameinlink]{cleveref} 

\newcommand{\smallcap}[1]{\textsc{#1}\index{#1@$\textsc{#1}$}}
\newcommand{\pad}{\smallcap{pad}}

\newcommand{\ib}{\mathsf{FP}^{\mathsf{PSPACE}}}
\newcommand{\iib}{\mathsf{IB}}
\newcommand{\Turing}{\mathrm{T}}
\newcommand{\manyone}{\mathrm{M}}
\newcommand{\noredux}{-}
\newcommand{\reversible}{\mathrm{r}}
\newcommand{\invertible}{\mathrm{i}}
\newcommand{\bijective}{\mathrm{b}}

\title{The Complexity of Iterated Reversible Computation}
\ThCSauthor{David Eppstein}{eppstein@ics.uci.edu}

\ThCSaffil{Department of Computer Science, University of California, Irvine}

\ThCSyear{2023}
\ThCSarticlenum{10}
\ThCSdoi{https://doi.org/10.46298/theoretics.23.10}
\ThCSreceived{Feb 10, 2022}
\ThCSrevised{Sep 22, 2023}
\ThCSaccepted{Nov 27, 2023}
\ThCSpublished{Dec 18, 2023}
\ThCSkeywords{Reversible computing, Circuit complexity, Cellular automata, Reversible cellular automata, Interval exchange transformations}

\ThCSshortnames{D.\ Eppstein}
\ThCSshorttitle{The Complexity of Iterated Reversible Computation}

\ThCSthanks{This research was supported in part by NSF grant CCF-2212129.}

\crefname{observation}{observation}{observations}
\begin{document}
\maketitle  

\begin{abstract}
We study a class of functional problems reducible to computing $f^{(n)}(x)$ for inputs $n$ and $x$, where $f$ is a polynomial-time bijection. As we prove, the definition is robust against variations in the type of reduction used in its definition, and in whether we require $f$ to have a polynomial-time inverse or to be computable by a reversible logic circuit. These problems are characterized by the complexity class $\ib$, and include natural $\ib$-complete problems in circuit complexity, cellular automata, graph algorithms, and the dynamical systems described by piecewise-linear transformations.
\end{abstract}

\section{Introduction}

\emph{Reversible logic} circuits, made from Boolean gates with bijective input-output mappings, 
can simulate any other combinational logic circuit~\cite{Ben-IBM-73}.  Other natural models of computing based on low-level reversible operations include reversible cellular automata~\cite{Kar-LS-92,Kar-UC-09,Tof-JCSS-77,Mor-TCS-95}, reversible Turing machines~\cite{JacMenSon-SIDMA-90,Cer-TMMO-60,Lec-CR-63}, and even reversibility in LISP programs, where the pair of fundamental operations (car,cdr) can be seen as inverse to cons~\cite{Kor-IJCAI-81,Epp-IJCAI-85}.
Early motivation for reversible computing came from a lower bound on the energy needed for each transition of a conventional logic gate, arising from fundamental physics~\cite{Lan-IBM-61}, and the realization that reversible devices can transcend this bound and compute with arbitrarily small amounts of energy~\cite{Ben-IJTP-82,Vit-CF-05}. More recently, reversible computing has gathered interest and attention from the recognition that quantum circuits must be reversible~\cite{Per-PRA-85,Deu-PRSA-89}. Deterministic reversible logic gates do not have the full power of quantum logic, but are the only kind of deterministic gates that can be incorporated into quantum circuits, so understanding their power is important for understanding the power of quantum computing more generally~\cite{tHo-JSP-88}.

Combinational logic circuits are acyclic; the sequential logic circuits of classical computers form cycles by feeding the output of a combinational logic circuit back to its input. One way to model this kind of feedback is through functional iteration: If $f$ is a function with equally many inputs and output bits, then $f^{(n)}(x)$ denotes the result of repeatedly applying $f$, $n$ times, to an input value~$x$. When a function $f$ is described as a combinational logic circuit, computing this kind of iterated function, for large values of $n$, requires the full power of $\mathsf{PSPACE}$. In one direction, the space needed to directly simulate the iterated computation of $f$, storing the circuit, its values, and an iteration counter, is polynomial. And in the other direction, any $\mathsf{PSPACE}$ computation can be simulated by the iteration of a circuit whose input and output bits describe the state of a space-bounded Turing machine, and whose evaluation takes one state to the next. However, $\mathsf{PSPACE}$-completeness does not imply that repeated computation of $f$ is the only way to compute $f^{(n)}$: for some computationally universal problems of this flavor, such as simulating the behavior of Conway's Game of Life, many instances can be simulated in time $o(n)$, at least in practice if not in theory~\cite{Gos-PD-84}.

The combination of these two ideas, that reversible logic can simulate combinational logic, and that iterated combinational logic characterizes $\mathsf{PSPACE}$, would suggest that iterated reversible logic might again be as powerful as $\mathsf{PSPACE}$. However, this is not obvious and does not follow from the known simulations of combinational logic by reversible logic. The problem is that  these simulations use additional bits of input and output, beyond those of the simulated circuit. Padding bits of input (assumed to be zero) are transformed by the reversible circuits of these simulations into garbage bits of output (not necessarily zero)~\cite{Ben-IBM-73}. In iterating these reversible circuits, the garbage bits may build up to more than polynomial space. What then is the computational power of iterated reversible computation? It is this question that we investigate in this paper.

To do so, we need to carefully define reversibility. We are interested in the iteration of bijections from $n$-bit inputs to $n$-bit outputs. However, the bijections that we iterate can be characterized computationally in at least three natural ways: we can assume that they are described by reversible circuits (as above), we can assume that they are given merely as polynomial-time functions, or we can assume that they are \emph{invertible}, polynomial-time bijections whose inverse function is also a polynomial-time bijection (as it would be for a function described by a reversible circuit). The question of whether all polynomial-time bijections are invertible is a major open problem in the complexity of reversible computing~\cite{JacMenSon-SIDMA-90}. If they are all invertible, it would imply the non-existence of one-way permutations, deterministic polynomial-time bijections whose inverse is not computable in randomized polynomial time, even in the average case rather than the worst case~\cite{Yao-FOCS-82,ImpRud-STOC-89}. In the other direction, a proof that one-way permutations exist would have significant cryptographic applications, but would also imply that $\mathsf{P}\ne\mathsf{NP}$, currently out of reach. The distinction between these three types of bijection, in the context of function iteration, can be seen as a space-complexity analogue of the same open problem, one that we resolve in this work.

We are led by these considerations to the following class of computational problems:

\begin{quotation}
\noindent
Let $f$ be a polynomial-time function whose numbers of input and output bits are equal, and assume in addition
that one of the following is true:
\begin{itemize}
\item $f$ is a bijection from its inputs to its outputs,
\item $f$ is a bijection whose inverse function is also computable in polynomial time, or
\item $f$ can be computed by a uniform family of reversible logic circuits.
\end{itemize}
What is the complexity of computing $f^{(n)}(x)$ from $n$ and $x$?
\end{quotation}

As we prove, the three variations in assumptions about $f$ lead to equivalent complexity for this problem, even though they may describe different classes of functions. To show this, we formulate a complexity class of functional problems encapsulating the  iteration of polynomial-time bijections. As we prove, the resulting complexity class turns out to equal $\ib$, the class of functional problems that can be solved in polynomial time with access to a $\mathsf{PSPACE}$ oracle.\footnote{Note that $\ib\ne\mathsf{FPSPACE}$, the class of functional problems solvable in polynomial space, because $\mathsf{FPSPACE}$ (as defined e.g. by Ladner~\cite{Lad-SICOMP-89}) does not count its output against its space complexity, and includes problems with exponentially long outputs. In contrast $\ib$ is limited to outputs of polynomial size.} We demonstrate the applicability of this theory by finding concrete $\ib$-complete problems coming from circuit complexity, graph algorithms, cellular automata, and dynamical systems.

\subsection{New results}

We prove the following results.
\begin{itemize}
\item We define a family of functional complexity classes based on the iteration of bijective polynomial-time functions, with nine variations depending on whether we incorporate reductions from other problems into the definition, which kind of reduction is allowed, and which specific class of polynomial-time bijection is allowed (\cref{def:9}). Despite this high variation, we show that six of these classes, the ones allowing either Turing or many-one reductions and allowing any of three types of polynomial-time bijection, are equal to each other (\cref{thm:6way}).
\item We observe that these equivalent complexity classes have a complete problem in circuit complexity, falling naturally out of one of our equivalent formulations: given a reversible logic circuit, a number $n$, and an initial value for the circuit input wires, what is the result of feeding the outputs of the circuit back through the inputs for $n$ iterations through the circuit? Equivalently, what would be the output of the circuit formed by composing in sequence $n$ copies of the given circuit? (\cref{obs:circuit-completeness}.)
\item We consider a family of computational problems on implicit graphs (graphs defined procedurally by polynomial-time algorithms for listing the neighbors of each vertex) in which the graph is undirected of maximum degree two, the input is a leaf (a vertex of degree one), and the desired output is the other leaf in the same connected component. Such ``second-leaf problems'' are known to be hard for $\ib$~\cite{Ben-84,Pap-JCSS-94,LanMcKTap-JCSS-00}, and include producing the same result as Thomason's lollipop algorithm for a second Hamiltonian cycle in a cubic graph~\cite{Tho-ADM-78}. We show that all second-leaf problems can be transformed into an equivalent iterated bijection (\cref{thm:leaf-in-bif}), and as a consequence that the complexity class of iterated bijections, under any of its six equivalent definitions, equals $\ib$ (\cref{thm:FPPS}).
\item We study finding the configuration of a reversible cellular automaton, after $n$ steps from a given initial configuration. We show that this is complete for $\ib$ for the billiard-ball model, a two-dimensional reversible cellular automaton of Margolus~\cite{Mar-PD-84} (\cref{thm:bbm-hardness}).
Although certain one-dimensional reversible cellular automata were known to be Turing-complete, their proofs of universality cannot be used for $\ib$-completeness, because they use space proportional to time.  Instead, we find a new family of one-dimensional reversible cellular automata that can simulate any two-dimensional reversible cellular automaton with the Margolus neighborhood and suitable boundary conditions (\cref{thm:dim-redux}). Finding the configuration after $n$ steps for these one-dimensional automata is $\ib$-complete (\cref{thm:1d-rca-completeness}).
\item We consider a family of polynomial-time bijections defined by piecewise linear transformations, and we show that finding their iterated values is $\ib$-complete (\cref{thm:plb-complete}). However, a natural special case of this problem has a non-obvious polynomial time algorithm, obtained by a transformation into computational topology (\cref{thm:integer-exchange}).
\end{itemize}

In these problems, as is standard in computational complexity theory, all input values are assumed to be represented as binary sequences, and all time bounds are based on the Turing machine model or on polynomially-equivalent models such as random access machines with logarithmically-bounded word sizes. In particular, our model of computation does not allow unit-cost arithmetic operations on arbitrarily large numbers or exact real-number arithmetic.

\subsection{Related work}

Our research combines research from structural complexity theory, circuit complexity theory, the theory of cellular automata, graph algorithms, and dynamical systems theory, which we survey in more detail in the relevant sections and summarize here.

In structural complexity theory, many early complexity classes such as $\mathsf{P}$, $\mathsf{PSPACE}$, and $\mathsf{NP}$ were defined by bounding resources such as space or time in computational models such as deterministic or nondeterministic Turing machines. More recently, it has become common instead to see complexity classes defined by reducibility to certain fundamental problems or classes of problems: for instance, $\mathsf{PPA}$ and $\mathsf{PPAD}$ are based on searching for specific structures in graphs~\cite{Pap-JCSS-94}, while $\exists\mathbb{R}$ is based on reducibility to problems in the existential theory of the real numbers~\cite{Sch-GD-09}. Similarly, one can redefine $\mathsf{NP}$ by reducibility to brute force search algorithms on deterministic machines. Our work takes inspiration from this shift in perspective, and similarly describes a class of problems that are reducible to iteration of bijections. One difference, however, is that while $\mathsf{PPA}$, $\mathsf{PPAD}$, and $\exists\mathbb{R}$ are all in some sense ``near'' $\mathsf{NP}$, the resulting class is more similar to $\mathsf{PSPACE}$.

In the theory of reversible circuits, a central result is the ability of these circuits to simulate non-reversible combinational logic circuits~\cite{Ben-IBM-73}. Reversible logic gates or families of gates, such as the Fredkin gate~\cite{FreTof-IJTP-82} or the Toffoli gate~\cite{Tof-ICALP-80}, are said to be universal when they can be used for these simulations. This simulation is always possible when additional padding bits are included in the inputs and outputs of the simulated circuits. Without padding, it is not possible for these gates to compute all functions or even all bijective functions. Fredkin gate circuits can only compute functions that preserve Hamming weight~\cite{FreTof-IJTP-82}, and some easily-computed bijections cannot be computed by any reversible logic circuit with gates of bounded complexity (\cref{obs:non-universal}).

Reversible circuits and reversible cellular automata are connected through the simulation of Fredkin-gate circuits by Margolus's billiard-ball block-cellular automaton~\cite{Mar-PD-84}. This provides one route to Turing completeness of reversible cellular automata, for which many more constructions are known~\cite{Tof-JCSS-77,Mor-TCS-95,Dur-STACS-97,ImaMor-TCS-00,KarSalTor-RC-14,MilFre-CF-05}. Although Turing completeness is closely related to the completeness properties studied here, these constructions generally involve infinite arrays of cells and in some cases initial conditions with infinite support, in contrast to our focus on space-bounded complexity. Another route to Turing completeness is the simulation of other cellular automata by reversible cellular automata. The billiard-ball model can simulate any other two-dimensional locally reversible cellular automata~\cite{Dur-CBC-02}, and Toffoli simulates arbitrary (non-reversible) $d$-dimensional cellular automata by $(d+1)$-dimensional reversible cellular automata~\cite{Tof-JCSS-77}. Our construction of a one-dimensional $\ib$-complete (and Turing-complete) reversible cellular automaton that simulates a two-dimensional one stands in sharp contrast to a result of Hertling~\cite{Her-UMC-98} that (under weak additional assumptions) simulations of cellular automata by reversible ones must increase the dimension, as Toffoli's construction does. Our construction does not meet Hertling's assumptions.

Another well-established connection relates algorithmic problems on implicitly-defined graphs to problems in computational complexity, by considering graphs that describe the state spaces of Turing machines or other computational models. It is standard, for instance, to reinterpret Savitch's theorem relating nondeterministic and deterministic space complexity as  providing a quadratic-space algorithm for reachability in implicit directed graphs. Similarly, the Immerman--Szelepcs\'enyi theorem on closure of nondeterministic space classes under complement has an equivalent algorithmic form, a nondeterministic linear-space algorithm for non-reachability in directed graphs~\cite{Wig-MFCS-92}. The complexity classes $\mathsf{PPA}$ and $\mathsf{PPAD}$ were formulated in the same way from algorithmic problems on implicit graphs~\cite{Pap-JCSS-94}. A specific algorithm that has been frequently studied in this light is Thomasen's lollipop algorithm for finding a second Hamiltonian cycle in an (explicit) cubic graph by following a path in a much larger implicit graph defined from the given graph~\cite{Tho-ADM-78}. Although some inputs cause this algorithm to take exponential time~\cite{BriSza-DM-22,Cam-DM-01,Zho-BAMS-18} the complexity of finding a second Hamiltonian cycle in a different way is unknown, and was one of the motivating problems for the definition of $\mathsf{PPA}$~\cite{Pap-JCSS-94}. We formulate the same question in a different way, asking how hard it is to find the same cycle that Thomasen's algorithm finds, but again the complexity of this problem remains unknown.

An important precursor of our work is the result of Bennett~\cite{Ben-84} and of Lange, McKenzie, and Tapp~\cite{LanMcKTap-JCSS-00} that reversible Turing machines with polynomial space can compute functions complete for $\ib$. Citing Bennett, Papadimitriou~\cite{Pap-JCSS-94} rephrased this result in terms of implicit graphs: it is complete for $\ib$ to find the other end of a path component in an implicit graph, given a vertex at one end of the path. Our proof that the iterated functional problems we study are complete for $\ib$ is based on this result, and on a reduction converting this path problem into an iterated bijection. For related time-space tradeoffs in the power of reversible Turing machines, see also Williams~\cite{Wil-00}.

Our final section concerns the iteration of invertible piecewise linear functions. This topic is well studied in the theory of dynamical systems; previously studied functions of this type include Arnold's cat map~\cite{DysFal-AMM-92} and the baker's map~\cite{Orn-AMM-71}, both acting on the unit square, and the interval exchange transformations on a one-dimensional interval~\cite{Kea-MZ-75}. The focus of past works on these transformations has been on their chaotic dynamics, rather than on the computational complexity of computing their iterates. We also consider perfect shuffle permutations, formulated as piecewise linear functions; their iterates have again been studied, notably to determine their order in the symmetric group~\cite{DiaGraKan-AAM-83}.

\section{Invertability, reversibility, and reversible logic}
\label{sec:bir}

We define a \emph{bijection of bitstrings} to be a function on binary strings of arbitrary length that, on $n$-bit inputs, produces $n$-bit outputs, and is one-to-one for each $n$. We define a \emph{polynomial-time bijection} to be a bijection of bitstrings computable in polynomial time, and we define a \emph{polynomial-time invertible bijection} to be a polynomial-time bijection whose inverse function is also a polynomial-time bijection. We define a \emph{reversible logic gate} to be a Boolean logic gate with equally many input and output bits that computes a bijection from inputs to outputs; the number of inputs and outputs is its \emph{arity}. Finally, we define a \emph{polynomial-time reversible function} to be a bijection of bitstrings that, for each $n$, can be computed by a circuit of reversible logic gates of fixed arity that can be constructed from the argument $n$ in time polynomial in $n$. Any polynomial-time reversible function or its inverse can be computed in polynomial time by constructing and simulating its circuit, so the polynomial-time reversible functions are a subset of the polynomial-time invertible bijections.

As in classical logic, certain reversible gates have been identified as \emph{universal}. These include the three-input \emph{Fredkin gate} in which one \emph{control} input is passed through unchanged but determines whether to swap the other two inputs~\cite{FreTof-IJTP-82}, and the \emph{Toffoli gate} in which the conjunction of two control inputs determines whether to negate a third input~\cite{Tof-ICALP-80}. Universality, in this context, has sometimes been incorrectly stated as meaning that all bijections can be implemented with these gates. This is impossible for any finite set of gates:

\begin{observation}
\label{obs:non-universal}
Let $n>1$ and let $-$ be the polynomial-time invertible bijection that takes an $n$-bit binary number $x$ in 2's-complement notation to its negation $-x$, with zero and the all-ones binary value taken to themselves. Then $-$ is not a polynomial-time reversible function. More strongly, no $n$-input $n$-output reversible logic circuit with gates of arity less than $n$ can compute~$-$.
\end{observation}

\begin{proof}
Given any reversible circuit, match up inputs and outputs at each gate to form $n$ longer wires running through the entire circuit, and describe the function of the circuit as a permutation on the $2^n$ truth assignments to these $n$ wires, obtained by composing in topological order permutations at each gate. Each gate's permutation is even, because each cycle in its permutation is paired with another cycle, obtained by negating a value unused by the gate. Therefore, the function of the whole circuit is also an even permutation. However, binary negation swaps $(2^n-2)/2$ pairs of values, an odd number for $n>1$, giving an odd permutation.
\end{proof}

We do not expect permutation parity to be the only obstacle to the existence of reversible circuits. For instance, modifying a function by adding an extra input that is passed unchanged to the output and otherwise does not affect the result (unlike a padding bit, which must be zero for correct output) can change the permutation parity from odd to even, but appears unlikely to affect the existence of a circuit, although we do not prove this.
Nevertheless, every logic circuit (even an irreversible one) can be simulated by a reversible circuit of approximately the same size with more inputs and outputs (equally many total inputs and outputs). The added ``dummy'' inputs must all be set to zero in the simulation, and the added ``garbage'' outputs produce irrelevant values, discarded in the simulated output. Each gate of the simulated circuit can be transformed into $O(1)$ reversible gates that use up $O(1)$ dummy inputs and produce $O(1)$ garbage outputs~\cite{Ben-IBM-73}. A stronger version of this simulation, for polynomial-time invertible bijections, uses the same zero-input padding but produces zeros for the output bits instead of garbage, allowing the resulting functions to be iterated. This is, essentially, a result of Jacopini, Mentrasti, and Sontacchi~\cite{JacMenSon-SIDMA-90}, but we include the proof because it is central to our later results, because we use similar techniques in other proofs, and because Jacopini, Mentrasti, and Sontacchi phrased it in terms of reversible Turing machines rather than reversible logic circuits. Define $\pad(k,x)$ to be the result of padding a binary string $x$ by prepending $k$ zero-bits.

\begin{lemma}[Jacopini, Mentrasti, and Sontacchi~\cite{JacMenSon-SIDMA-90}]
\label{lem:invertible-to-reversible}
Let $f$ be a polynomial-time invertible bijection. Then there exists a polynomial-time reversible function $g$, and a polynomial $p$, such that for every binary string $x$ of length $n$,
\[g\Bigl(\pad\bigl(p(n),x\bigr)\Bigr)=\pad\bigl(p(n),f(x)\bigr).\]
That is, on strings consisting of $p(n)$ zeros followed by an $n$-bit string $x$, $g$ behaves like the evaluation of $f$ on the final $n$ bits, leaving the zeros unchanged.
\end{lemma}

\begin{proof}
Construct a circuit on padded inputs that performs the following computations:
\begin{enumerate}
\item Replace $n$ padding bits with their bitwise exclusive or with input $x$,
transforming a padded input of the form $\bar 0,\bar 0,x$ (where the first $\bar 0$ denotes the remaining unused padding bits, the second $\bar 0$ denotes the padding bits replaced in this computation, and the comma denotes concatenation)
into $\bar 0,x,x$, the remaining unused padding bits together with two copies of~$x$.
\item\label{forward-step} Expand the polynomial-time computation of $f$, on $n$-bit inputs, into a classical logic circuit, and
use the simulation of classical logic by reversible logic to compute $f$ on one of the copies of $x$, replacing more padding bits by garbage bits. After this step, the input has been transformed into the form $y,f(x),x$, where $y$ is the garbage produced by the simulation.
\item Use additional bitwise exclusive ors to transform the input to the form $y,f(x),f(x)\oplus x$.
\item Reverse the circuit of Step~\ref{forward-step} to transform the input to the form $\bar 0,x,f(x)\oplus x$.
\item Use additional bitwise exclusive ors to transform the input to the form $\bar 0,f(x),f(x)\oplus x$.
\item\label{backward-step} Use the simulation of classical by reversible logic to compute $f^{-1}$ on input $f(x)$, transforming the input to the form $z,x,f(x)\oplus x$ where $z$ is the garbage produced by the simulation.
\item Use additional bitwise exclusive ors to transform the input to the form $z,x,f(x)$.
\item Reverse the circuit of Step~\ref{backward-step} to transform the input to the form $\bar 0,f(x),f(x)$.
\item Use additional bitwise exclusive ors to transform the input to the form $\bar 0,\bar 0,f(x)$.\qedhere
\end{enumerate}
\end{proof}

This proof produces circuits that combine Fredkin or Toffoli gates with additional two-input two-output exclusive or gates (also called controlled not gates). If a circuit using only one type of gate is desired, then these controlled not gates can be simulated by Fredkin gates or Toffoli gates by using additional (reusable) dummy bits that, like the padding bits, can be passed unchanged from the input to the output of the resulting circuit. We omit the details.

\section{Complexity classes and their equivalences}

We will consider two different types of reduction in our definitions of completeness:
\begin{itemize}
\item A \emph{polynomial-time functional Turing reduction} (Turing reduction, for short) from one functional problem $f$ to another functional problem $g$ can be described as a polynomial-time oracle Turing machine for problem $X$, using an oracle for problem $Y$. That is, it is an algorithm for computing the function $f$ that is allowed to make subroutine calls to an algorithm for function $g$, and that takes polynomial time outside of those calls.
\item A \emph{polynomial-time functional many-one reduction} (many-one reduction, for short) consists of two polynomial-time algorithms $r_1$ and $r_2$, such that $f=r_1\circ g\circ r_2$. That is, we can compute~$f$ by translating its input in polynomial time into an input for function $g$, computing a single value of $g$, and then translating the computed value of $g$ in polynomial time into the value of $f$. Equivalently, this can be thought of as a Turing reduction that is limited to a single oracle call.
\end{itemize}
Additionally, as we have already seen, we have three choices of which type of bijection to use in the iteration.
This naturally gives rise to the nine variant complexity classes defined below. However, we will later see that six of these are actually the same as each other (and all the same as the known complexity class $\ib$): as long as the definition of complexity class includes one of these two types of reduction, the choice of reduction type and bijection type does not matter.

\begin{definition}
\label{def:9}
Define the nine complexity classes $\iib_{x,y}$, for $x\in\{\Turing,\manyone,\noredux\}$ and $y\in\{\bijective,\invertible,\reversible\}$ (short for ``iterated bijection''), as follows.
\begin{itemize}
\item The complexity classes $\iib_{\noredux,\bijective}$, $\iib_{\noredux,\invertible}$, and $\iib_{\noredux,\reversible}$ denote the classes of problems for which the input is a pair $(n,s)$ and the output is $f^{(n)}(s)$, where the $n$-times iterated function $f$ is respectively a polynomial-time bijection, a polynomial-time invertible bijection, or a polynomial-time reversible function.

\item The complexity classes $\iib_{\Turing,y}$ denote the classes of problems having a polynomial-time functional Turing reduction to a problem in $\iib_{\noredux,y}$, for each $y$ in $\{\bijective,\invertible,\reversible\}$.

\item The complexity classes $\iib_{\manyone,y}$ denote the classes of problems having a polynomial-time functional many-one reduction to a problem in $\iib_{\noredux,y}$, for each $y$ in $\{\bijective,\invertible,\reversible\}$.
\end{itemize}
\end{definition}

Padding an input by zeros and unpadding the output in the same way is a many-one reduction, and every many-one reduction is also a Turing reduction. For both kinds of reduction, the composition of two reductions is another reduction.
Therefore, the following is an immediate consequence of \cref{lem:invertible-to-reversible}, according to which every polynomial-time invertible bijection can be padded to an equivalent polynomial-time reversible function.

\begin{observation}
\label{obs:invertible-to-reversible}
$\iib_{\Turing,\invertible}=\iib_{\Turing,\reversible}$ and
$\iib_{\manyone,\invertible}=\iib_{\manyone,\reversible}$.
\end{observation}

Next, we show that any Turing reduction to iterating a bijective or invertible function can be strengthened to a many-one reduction of a different function in the same class. That is:

\begin{lemma}
\label{lem:reduction-equivalence}
$\iib_{\Turing,\bijective}=\iib_{\manyone,\bijective}$ and
$\iib_{\Turing,\invertible}=\iib_{\manyone,\invertible}$.
\end{lemma}

\begin{proof}
Let $f$ be a functional problem in $\iib_{\Turing,\bijective}$ or $\iib_{\Turing,\invertible}$. This means that $f$ can be solved by an algorithm $A_f$, that takes polynomial time outside of a polynomial number of calls to a subroutine for computing $g^{(n)}(s)$ for a fixed polynomial-time bijection $g$. To show that $f\in\iib_{\manyone,\bijective}$ (or, respectively, $f\in\iib_{\manyone,\invertible}$), we construct a different polynomial-time bijection $h$ whose iteration will simulate the behavior of algorithm $A_f$. In preparation for doing so, we expand $A_f$ into a (conventional logic) circuit of polynomial size, consisting of the standard Boolean logic gates together with a special many-input many-output gate that takes as input $n$ and $s$ and produces as output $g^{(n)}(s)$, implementing the oracle calls of algorithm~$A_f$. We will simulate this circuit gate-by-gate, in a topological ordering of its gates, by a function $h$ that operates on triples $(c_1,c_2,b)$, where:
\begin{itemize}
\item The value $c_1$ (the ``big hand of the clock'') will indicate the progression of the simulation through the gates of the expanded circuit for algorithm~$A_f$.
\item The value $c_2$ (the ``little hand of the clock'') will indicate the progression of the simulation through an iteration of function $g$, within a single oracle gate of the circuit.
\item The value $b$ will indicate the Boolean values on all wires of the circuit, with zeros for wires whose value has not yet been determined by the simulation.
\end{itemize}
Initially, these values will all be zero, except for the values in $b$ that describe input wires of the simulated circuit; the part of the many-one reduction that determines the initial value of the iterated function can easily calculate what these input wire values should be.

The function $h$ that performs a step of the simulation will always increase $c_2$ by one modulo a suitable value $N$, and if the result is zero it will increase $c_1$ by one modulo a suitable value $M$. These moduli are chosen so that $N$ is larger than the largest possible argument $n$ in an oracle call to $g^{(n)}(s)$, and so that $M$ is larger than the number of gates in the simulated circuit. Because these increments are performed in modular arithmetic, they are bijective and invertible. We will iterate $h$ for $MN$ iterations, so that the big hand will increase for at least as many steps as the number of gates to be simulated. Each iteration of $h$ will also perform additional invertible functions, depending on $c_1$ and $c_2$:
\begin{itemize}
\item If $c_1$ is the position of standard logic gate $G$ in the topological ordering of the circuit for algorithm~$A_f$, and $c_2=0$, then let $o$ be the correct output of $G$, let $b_i$ be the wire where that output should go, and let $h$ replace $b_i$ by its exclusive or with $o$. This operation is bijective and invertible (it is its own reverse).
\item If $c_1$ is the position of a standard logic gate $G$, and $c_2\ne 0$, then $h$ does nothing beyond incrementing its counters.
\item If $c_1$ is the position in the topological ordering of an oracle gate with input $n,s$ and output~$t$, computing $t=g^{(n)}(s)$, and $c_2=0$, then $h$ uses bitwise exclusive ors (as in the proof of \cref{lem:invertible-to-reversible}) to copy $s$ onto $t$.
\item If $c_1$ is the position in the topological ordering of an oracle gate with input $n,s$ and output~$t$, and $0<c_2\le n$, then $h$ replaces $t$ by $g(s)$. If $g$ is bijective, this operation is bijective, and if~$g$ is invertible, this step is invertible.
\item If $c_1$ is the position in the topological ordering of an oracle gate with input $n,s$ and output~$t$, and $c_2>n$, then $h$ does nothing beyond incrementing its counters.
\end{itemize}
Finally, the part of the many-one reduction that maps the output of $h^{(MN)}$ to the value of $f$ does so simply by copying the output bits of the simulated circuit.
\end{proof}

The next observation reduces the computation of a polynomial time bijection (for which we do not necessarily have a polynomial-time inverse) to the iteration of a different polynomial-time invertible bijection (for which we do have the inverse). We include it here to introduce a counting trick in its proof, which we will use in a more complicated way in what follows.

\begin{observation}
\label{obs:inversion-by-iteration}
Let $f$ be a polynomial time bijection. Then $f\in\iib_{\manyone,\invertible}$.
\end{observation}

\begin{proof}
We reduce the computation of $f(x)$ to the iteration of an invertible function through the trivial summation
\[
f(x)=\sum_{y}
\begin{cases}
f(y) & \mbox{if }y=x\\
0 & \mbox{otherwise.}\\
\end{cases}
\]
To do so, define $\tilde f(y)$ to be $f(y)$, if $y=x$, and $0$, otherwise, so that the sum is over the values of $\tilde f$.
Create an invertible function $g$ that operates on pairs $a,b$,
and maps $(a,b)\mapsto \bigl(a+1\bmod N,b+\tilde f(a)\bigr)$, invertible via the map
$(a,b)\mapsto \bigl(a-1\bmod N,b-\tilde f(a-1)\bigr)$, where $N$ is the number of possible inputs to function $f$. Then we can simply evaluate $f(x)$ as the $b$-component of $g^{(N)}(0,0)$.
\end{proof}

To simulate the iteration of a polynomial-time bijection using invertible steps, we combine the trivial-summation idea of \cref{obs:inversion-by-iteration}, the alternating forward and backward steps of \cref{lem:invertible-to-reversible}, and the big-hand little-hand timing idea of \cref{lem:reduction-equivalence}, as follows.

\begin{lemma}
\label{lem:iterated-inversion}
$\iib_{\noredux,\bijective}\subset\iib_{\manyone,\invertible}$.
\end{lemma}

\begin{proof}
Let $f$ be a polynomial-time bijection for which we wish to compute $f^{(n)}(x)$, the form taken by all problems in $\iib_{\noredux,\bijective}$. We must show that $f^{(n)}(x)$ can be computed in $\iib_{\manyone,\invertible}$, by iterating a polynomial-time invertible bijection. To do so, we define  an invertible bijection $g$ on 5-tuples $(c_1,c_2,a,b,c)$ where $c_1$ and $c_2$ are the big hand and little hand of the big-hand little-hand timing technique, $p$ is an adequate supply of polynomially many padding bits (zero before and after each iteration), $a$ is the current iterated value (initially the starting value $x$), and $b$ and $c$ are equally-long values used within the iteration. If the inputs and outputs to $f$ have $k$ bits, we will choose the lengths of the values in these 5-tuples to all be monotonic and easily-computed functions of $k$, so that the computation of $g$ can determine $k$ and decode the 5-tuple to its components in polynomial-time; we omit the details of this decoding process. For inputs to $g$ whose length is not of the correct form to be decoded into a 5-tuple in this way, we define $g$ to be the identity function.

Otherwise, as in \cref{lem:reduction-equivalence}, we define $g$ so that in each iteration it increments $c_2$ modulo some sufficiently large number $M$ and, if the resulting value of $c_2$ is zero, it also increments $c_1$ modulo some sufficiently large number $N>n$. Each increase of $c_1$ will correspond to one more iteration of the function $f$, so that $f^{(n)}(x)$ may be obtained by iterating $g$ exactly $nM$ times, starting from $(0,0,x,0,0)$, and examining the $a$ component of the resulting tuple. If $a$ has $k$ bits, $M$ is chosen to be greater than $2^k+2$. The effect of $g$ on the $a$, $b$, and $c$ components of the 5-tuple are determined by the value of $c_2$:
\begin{itemize}
\item If $c_2=0$, function $g$ sets $b=b\oplus f(a)$. Our overall simulation will only perform this step with $b$ and $c$ initially zero,
transforming $a,0,0$ to $a,f(a),0$, corresponding to Step 2 of \cref{lem:invertible-to-reversible}.
\item If $c_2=1$, function $g$ sets $a=a\oplus b$. When applied to $a,f(a),0$, this transforms it into $a\oplus f(a),f(a),0$, corresponding to Step 3 of \cref{lem:invertible-to-reversible}.
\item If $1<c_2<2^k+2$, function $g$ checks whether $f(c)=b$, and if so replaces $a$ with $a\oplus c$. Then regardless of the outcome of the check, it increments $c$ modulo $2^k$. When applied to $a\oplus f(a),f(a),0$, these steps use the trivial summation method to find $a$ and exclusive-or it into the first component, producing $f(a),f(a),0$ as in Step 8 of \cref{lem:invertible-to-reversible}.
\item If $c_2=2^k+2$, function $g$ replaces $b$ by $a\oplus b$. When applied to $f(a),f(a),0$, this produces $f(a),0,0$, ready for another iteration.
\item For all other values of $c_2$, $g$ does nothing to $a$, $b$, and $c$.
\end{itemize}
The changes to $c_1$ and $c_2$ in each computation of $g$ are easily inverted, and other than those changes the only effect of $g$ is to perform an exclusive-or into one of the three components $a$, $b$, and $c$, with a value computed only from the other two components, an operation that is its own inverse. Therefore, $g$ is a polynomial-time invertible bijection, and we have shown how to compute the iterated values of a polynomial-time bijection~$f$ by iterating a different polynomial-time invertible bijection~$g$.
\end{proof}

\begin{theorem}
\label{thm:6way}
For $x\in\{\manyone,\Turing\}$ and $y\in\{\bijective,\invertible,\reversible\}$ the complexity classes $\iib_{x,y}$ are all equal.
\end{theorem}

\begin{proof}
From their definitions, these classes are naturally partially ordered by inclusion, with $\iib_{\manyone,y}\subseteq\iib_{\Turing,y}$ and $\iib_{x,\reversible}\subseteq\iib_{x,\invertible}\subseteq\iib_{x,\bijective}$, so we need only show that every problem from the largest class in this partial order, $\iib_{\Turing,\bijective}$, is contained within the smallest class in this partial order, $\iib_{\manyone,\reversible}$. Therefore, let $F$ be a functional problem in $\iib_{\Turing,\bijective}$, meaning that it can be reduced by a Turing reduction to the iteration of a polynomial-time bijection $f$. By \cref{lem:iterated-inversion} and the composition of this Turing reduction with the many-one reduction of \cref{lem:iterated-inversion} to produce another Turing reduction, $F\in\iib_{\Turing,\invertible}$. By \cref{lem:reduction-equivalence}, $F\in\iib_{\manyone,\invertible}$. And by \cref{obs:invertible-to-reversible}, $F\in\iib_{\manyone,\reversible}$.
\end{proof}

Because of this equivalence, it is justified to drop the subscripts and use $\iib$ to refer to any of the six equivalent complexity classes of \cref{thm:6way}. We will later see (\cref{thm:FPPS}) that $\iib=\ib$. For now, we prove the easy direction of this equivalence.

\begin{observation}
\label{obs:fp-pspace}
$\iib\subseteq\ib$.
\end{observation}

\begin{proof}
Any function in $\iib$ can be computed by performing a polynomial-time reduction and then
repeatedly computing a polynomial-time function, using a counter to keep track of the number of iterations performed. The reduction, the computation of the function, and the counter all use only polynomial space. Therefore, the $i$th bit of output of a function in $\iib$ belongs to $\mathsf{PSPACE}$, and all polynomially-many bits of output can be obtained by using a polynomial number of calls to a $\mathsf{PSPACE}$ oracle to obtain the bits for each different value of~$i$.
\end{proof}

\begin{observation}
\label{obs:circuit-completeness}
Define $R$ be a functional problem whose input is a specification of a reversible logic circuit composed of universal reversible logic gates, a number $n$, and a Boolean assignment to each input wire of the circuit, and whose output is the result of applying the circuit~$n$ times, passing its outputs back to its inputs. Then $R$ is complete for $\iib$ under both many-one and Turing reductions.
\end{observation}

\begin{proof}
This is the definitional problem for $\iib_{\noredux,\reversible}$, and its hardness for $\iib$ follows from the composition of reductions with problems in that class. We must also show that this problem itself belongs to $\iib$, but this is easy: it is the problem of computing the $n$th iterate of an invertible polynomial-time function $f$ that takes as input a specification of a reversible circuit and an assignment to its input wires, and that produces as output the unchanged specification of the circuit and the assignment to its output wires obtained by simulating the circuit. The inverse of this function can be obtained by applying it to the reversal of the specified circuit.
\end{proof}

\section{Implicit linear forests}

An \emph{implicit graph} is a graph whose vertices are represented as binary strings of a given length, and whose edges are determined by a computational process (an oracle or subroutine for computing neighbors of each vertex) rather than being listed explicitly in an adjacency list or other data structure. These typically represent state spaces of computations or of combinatorial structures, and the use of implicit graphs is common in complexity theory. Savitch's theorem, for instance, can be interpreted as defining an algorithm for finding a path between two selected vertices in an implicit directed graph, in low deterministic space complexity~\cite{Wig-MFCS-92}.

The problems we have already considered can easily be reformulated in the language of implicit graphs: a bijection can be thought of as a directed graph with the values on which it operates as vertices, and with in-degree and out-degree both exactly one at each vertex. The iteration problem we have been considering, rephrased in this language, asks for the vertex that one would reach by following a path of length $n$ in this graph. However, this is somewhat artificial as a graph problem. Instead, we consider undirected implicit graphs in which every connected component is a path, known as \emph{linear forests}, or more generally implicit graphs with maximum degree two. Given a leaf vertex (a vertex of degree one) in such a graph, how easy is it to find the other leaf of the same path? As we show in this section, this provides an alternative equivalent formulation of the class $\iib$ that is based on graph search rather than on bijective functional iteration.

\subsection{Thomason's lollipop algorithm}

\begin{figure}[t]
\centering\includegraphics[width=0.5\textwidth]{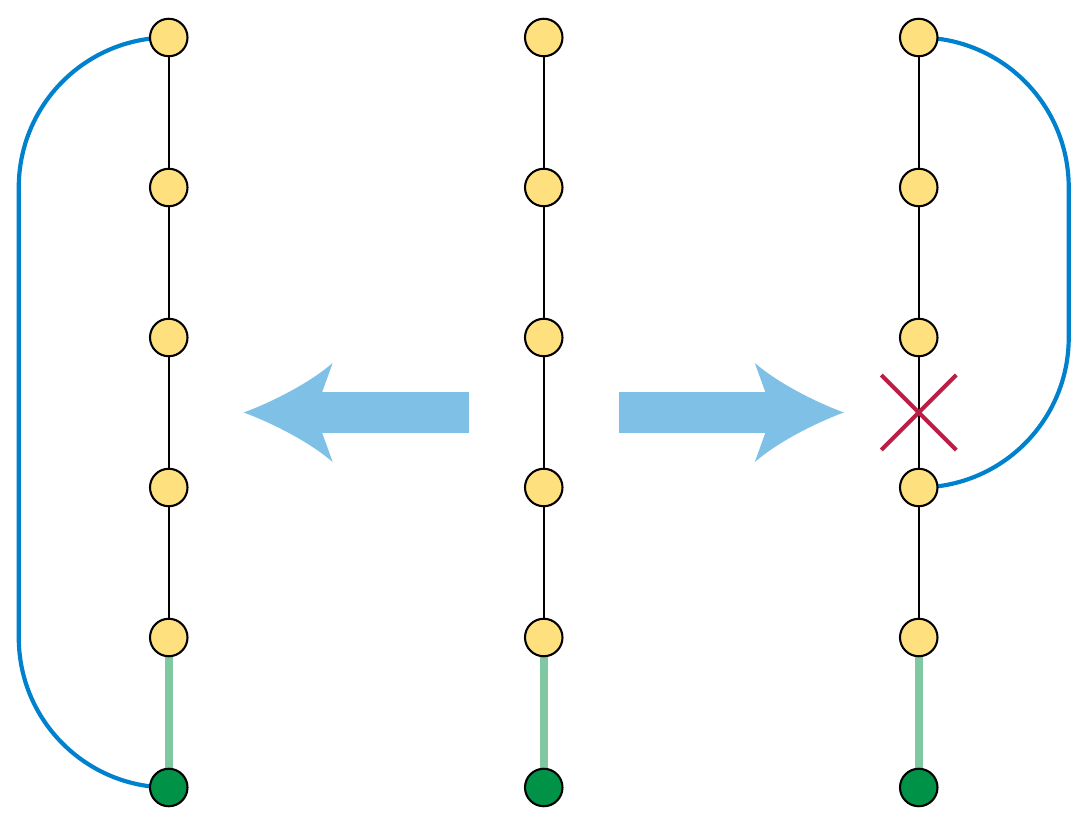}
\caption{State space and transitions for Thomason's lollipop algorithm. From any Hamiltonian path (center) with a fixed starting vertex and edge (green), extending the other end of the path by one more edge (blue) can either produce a Hamiltonian cycle (left) or a ``lollipop'', a shorter cycle with a dangling path (right). Removing one edge from the cycle in a lollipop (red X) produces another Hamiltonian path with the same fixed starting vertex and edge.}
\label{fig:lollipop}
\end{figure}

Before proving the equivalence of this formulation, we briefly discuss a prototypical example of a problem of this type, Thomason's lollipop algorithm for a second Hamiltonian cycle. In a 3-regular undirected graph, the number of Hamiltonian cycles through any fixed edge is even~\cite{Tut-JLMS-46}. A proof of this fact by Thomason~\cite{Tho-ADM-78} constructs a state space, or implicit graph, as follows (\cref{fig:lollipop}):
\begin{itemize}
\item The states of the state space are Hamiltonian paths starting at a fixed endpoint of a fixed edge. The initial Hamiltonian cycle can be transformed into one of these states by choosing arbitrarily one of its vertices and edges as the fixed vertex and edge of the state space, and removing the other Hamiltonian cycle edge that is incident to the chosen vertex.
\item Each state can transition to at most two other states, by adding one more edge to the far end of the Hamiltonian path from the fixed edge. The number of choices for this added edge is exactly two, because the given graph has degree three and one of the edges at the end vertex is already used as part of the path. If this edge is incident to the fixed vertex, adding it produces a Hamiltonian cycle; otherwise, adding it produces a ``lollipop'' or spanning subgraph in the form of a cycle with a dangling path. When it produces a lollipop, we can break the cycle at the other edge incident to the dangling path, and produce a new state. Therefore, the states that can form Hamiltonian cycles by the addition of an edge have exactly one neighbor, while the other states have exactly two neighbors.
\end{itemize}
The evenness of the number of Hamiltonian cycles through a fixed edge follows immediately from this construction: After fixing an orientation for the fixed edge, each Hamiltonian cycle corresponds to a degree-one state in this state space, which can only belong to a path of states. Every path has exactly two degree-one states, so the paths in the state space group the Hamiltonian cycles into pairs~\cite{Tho-ADM-78}.

The same argument also provides an algorithm for finding a second Hamiltonian cycle, given as input a single Hamiltonian cycle in an (explicitly represented) graph. One simply chooses arbitrarily an edge of this cycle to be a fixed edge and the orientation of this chosen edge, constructs the state space as above, and walks along the path in the state space from the initial state to another degree-one state, which must come from a different Hamiltonian cycle~\cite{Tho-ADM-78}.

Fnding a second Hamiltonian cycle is one of the prototypical examples of a problem in the complexity class $\mathsf{PPA}$, defined more generally in terms of finding a second odd-degree vertex in an implicit graph, although it is not known to be complete for $\mathsf{PPA}$~\cite{Pap-JCSS-94}. However, a solution of the $\mathsf{PPA}$ version of the problem is not required to be in the same component of the state space as the given Hamiltonian cycle, so $\mathsf{PPA}$ does not characterize the complexity of finding the same Hamiltonian cycle as the cycle found by Thomason's algorithm.
Instead, we are interested in the complexity of a more specific problem, solved by Thomason's algorithm: given a Hamiltonian cycle, a fixed edge, and a fixed orientation for that edge, find the other Hamiltonian cycle from the same component of Thomason's state space. This is an instance of the second leaf problem, in an implicit graph of maximum degree two (not necessarily a linear forest). It is known that some instances may cause Thomason's lollipop algorithm to take an exponential number of steps~\cite{Kra-JCSS-99,Cam-DM-01,Zho-BAMS-18,BriSza-DM-22}, but while this settles the complexity of this specific algorithm, it leaves open the complexity of the functional problem solved by the algorithm.

\subsection{Equivalence to iterated bijection}

We will formalize a class of computational problems like the one solved by Thomason's algorithm, rather than a single problem, in order to make the neighbor-finding subroutines by which we define an implicit graph be part of the problem definition rather than part of the input. However, we also need input data, used by those subroutines to specify the implicit graph. For instance, in the problem formalizing the input--output behavior of Thomason's lollipop algorithm, the definition of the problem includes the fact that its state space consists of Hamiltonian paths with fixed starts in a cubic graph, rather than being some other kind of implicit graph. However, the specific cubic graph containing these paths is input data rather than being part of the problem specification. Thus, we make the following definitions.

\begin{definition}
A \emph{parameterized family of implicit graphs} is defined by a polynomial time function $N(G,v)$ that takes as input two bitstrings $G$ and $v$, where $G$ identifies a specific implicit graph and $v$ names a vertex within that graph, and that produces as output a finite sequence of distinct bitstrings of equal length to $v$ naming the neighbors of $v$ in $G$. If $v$ is invalid (meaning that it does not name a vertex in the graph specified by $G$), $N$ should output a failure condition, again in polynomial time. A parameterized family is \emph{undirected} if, whenever $w$ belongs to the output of $N(G,v)$, $v$ symmetrically belongs to the output of $N(G,w)$. It is \emph{bivalent} if every call to $N$ produces at most two neighbors. A \emph{connected leaf problem} is defined by an undirected bivalent family of implicit graphs, defined by a polynomial time function $N$. An input to a problem defined in this way consists of input values $G$ and $v$ such that the output of $N(G,v)$ has length exactly one. The output to a connected leaf problem, defined in this way, is a bitstring $w$ describing a vertex of degree one in the same connected component as $v$ of the implicit graph defined by $N$ and $G$.
When the input does not have the correct form ($N(G,v)$ produces a failure condition or the wrong number of neighbors) the output is undefined.
\end{definition}

Thus, the problem of duplicating the output of Thomason's lollipop algorithm is a connected leaf problem in which $G$ describes the underlying cubic graph in which a second Hamiltonian cycle is to be found, $v$ encodes a description of a Hamiltonian path in this underlying graph, and~$N$ performs a single step of the lollipop algorithm described above, for each of the two ways of extending the path described by $v$, and outputs the Hamiltonian path or paths that result from this step.

\begin{theorem}
\label{thm:leaf-in-bif}
Every connected leaf problem belongs to $\iib$.
\end{theorem}

\begin{proof}
Given a connected leaf problem defined by a polynomial time function $N(G,v)$,
we find a Turing reduction from instances $(G,v)$ to equivalent problems of iterated bijection, as follows.
Let $k$ be the number of bits in the bitstring $v$; by the way we have defined parameterized families of implicit graphs, all vertices in the connected component of $v$ in the implicit graph specified by $G$ have the same number $k$ of bits in their descriptions. We construct a polynomial-time bijection $f$ whose inputs and outputs are triples $(n,v,w)$ of $k$-bit values. In these triples~$n$ can be interpreted as a number modulo $2^k$, at least as large as the number of vertices in the component of $v$. The two remaining values $v$ and $w$ in these triples should be interpreted as describing adjacent vertices in the graph specified by $G$. However, to formulate this as a problem in $\iib$, the bijection $f$ that we construct
cannot depend on this interpretation: it must be capable of handling values $v$ and $w$ that do not specify vertices, or that specify vertices that are not adjacent. We compute the value of $f$ according to the following case analysis:
\begin{itemize}
\item First, use $N$ to compute the neighbors $N(G,v)$ and $N(G,w)$ of $v$ and $w$. If either call returns a failure condition, or if the two vertices are not neighbors, return the input $(n,v,w)$ unchanged as the output.
\item If $n>0$, check whether $v$ has one neighbor. If so, return $(n+1\text{ mod }2^{k+1},v,w)$. However, if~$v$ has two neighbors, return $(n,v,w)$ unchanged.
\item In the remaining case, $n=0$. If $w$ has one neighbor, return $(1,w,v)$. Otherwise, $w$ has two neighbors, $v$ and another vertex $u$. In this case, return $(0,w,u)$.
\end{itemize}
If $v$ is a leaf vertex in the implicit graph described by $G$, and $w$ is its one neighbor, then iterating this function has the effect of walking in one direction along a path, waiting $2^k$ steps, and then walking in the same way in the opposite direction along the path, acting bijectively on all of the triples of values seen in this walk. If $v$ and $w$ are neighbors in a cycle of the implicit graph, then iterating this function starting from $(0,v,w)$ has the effect of walking from arc to arc around the cycle, with no waiting steps. For all of the remaining triples of values, this function acts as the identity. Therefore, in all cases it is bijective. Each step involves only two calls to the polynomial time function $N$, and simple case analysis, so it takes polynomial time to compute $f$.

We can solve a connected leaf problem with neighbor function $N$ and data $G,v$ by first using $N$ to find the neighbor $w$ of $v$ and then iterating the function $f$ constructed above for $2^k$ iterations starting from $(0,v,w)$ to produce another triple $(a,b,c)$, and finally returning~$b$. By the construction of $f$, iterating it will necessary reach the leaf at the other end of the component of $v$ in fewer than $2^k$ steps and then wait for $2^k$ steps while incrementing the first component of the triple modulo $2^k$ until it reaches zero again. If we iterate $f$ for exactly $2^k$ steps, the resulting triple $(a,b,c)$ will necessarily be part of this waiting stage of the dynamics of $f$, and the returned value $b$ will necessarily be the other leaf connected to $v$, as desired.
\end{proof}

\begin{theorem}
\label{thm:FPPS}
$\iib=\ib$.
\end{theorem}

\begin{proof}
This follows immediately from the fact that $\iib\subseteq\ib$ (\cref{obs:fp-pspace}), from \cref{thm:leaf-in-bif}, and from the known $\ib$-completeness of the connected leaf problem~\cite{Ben-84,Pap-JCSS-94,LanMcKTap-JCSS-00}.
\end{proof}

Because of this equivalence, from now on we will generally refer to this class by its conventional name, $\ib$, instead of the nonce name $\iib$.

Although producing the same output as Thomason's lollipop algorithm belongs to $\ib$, by \cref{thm:leaf-in-bif}, we do not know whether it is $\ib$-complete, just as we do not know whether finding an arbitrary second Hamiltonian cycle in a cubic graph is $\mathsf{PPA}$-complete.

\section{Reversible cellular automata}

A \emph{cellular automaton} has a finite set of states, and a periodic system of cells. For us, these cells will form one-dimensional or two-dimensional arrays; although it is common to treat these arrays as infinite, we will form finite computational problems by using arrays of varying size with periodic boundary conditions. A \emph{configuration} of the automaton assigns a state to each cell. The automaton is updated by simultaneously computing for each cell a new state, determined in a translation-invariant way as a function of the states of a constant number of neighboring cells. The resulting cellular automaton is \emph{reversible} if the transformation from one configuration to the next is a bijection. When a cellular automaton is reversible, its inverse transformation can also be described by a reversible cellular automaton~\cite{Hed-MST-69,Ric-JCSS-72}. A periodic array of reversible logic gates would define a reversible automaton whose reverse dynamics uses the same neighborhood structure, but other reversible cellular automata can have reverse neighborhoods that are much larger than the forward ones~\cite{Kar-LS-92}. Every one-dimensional or two-dimensional reversible cellular automaton can be defined by a rule with locally reversible steps, as would be obtained by an array of reversible gates, but for higher dimensions this remains unknown~\cite{Kar-UC-09}. Just as irreversible circuits can be simulated by reversible ones, irreversible cellular automata can be simulated by reversible ones at the cost of an increase in dimension~\cite{Tof-JCSS-77} or of the simulation becoming asynchronous~\cite{Mor-TCS-95}.

For a fixed reversible cellular automaton rule, a simulation should take as input an initial configuration $x$ and a number of steps $n$, and produce as output the configuration of the automaton after $n$ steps. We do not require this simulation to be performed by directly calculating the transformations from each configuration to the next; for instance, for the (non-reversible) Conway's Game of Life automaton, hashing techniques have been successful at running simulations using computation time substantially sublinear in the number of simulated steps~\cite{Gos-PD-84}. What is the complexity of simulating reversible cellular automata? 

\subsection{Billiard-ball model}

\begin{figure}[t]
\centering\includegraphics[width=0.8\textwidth]{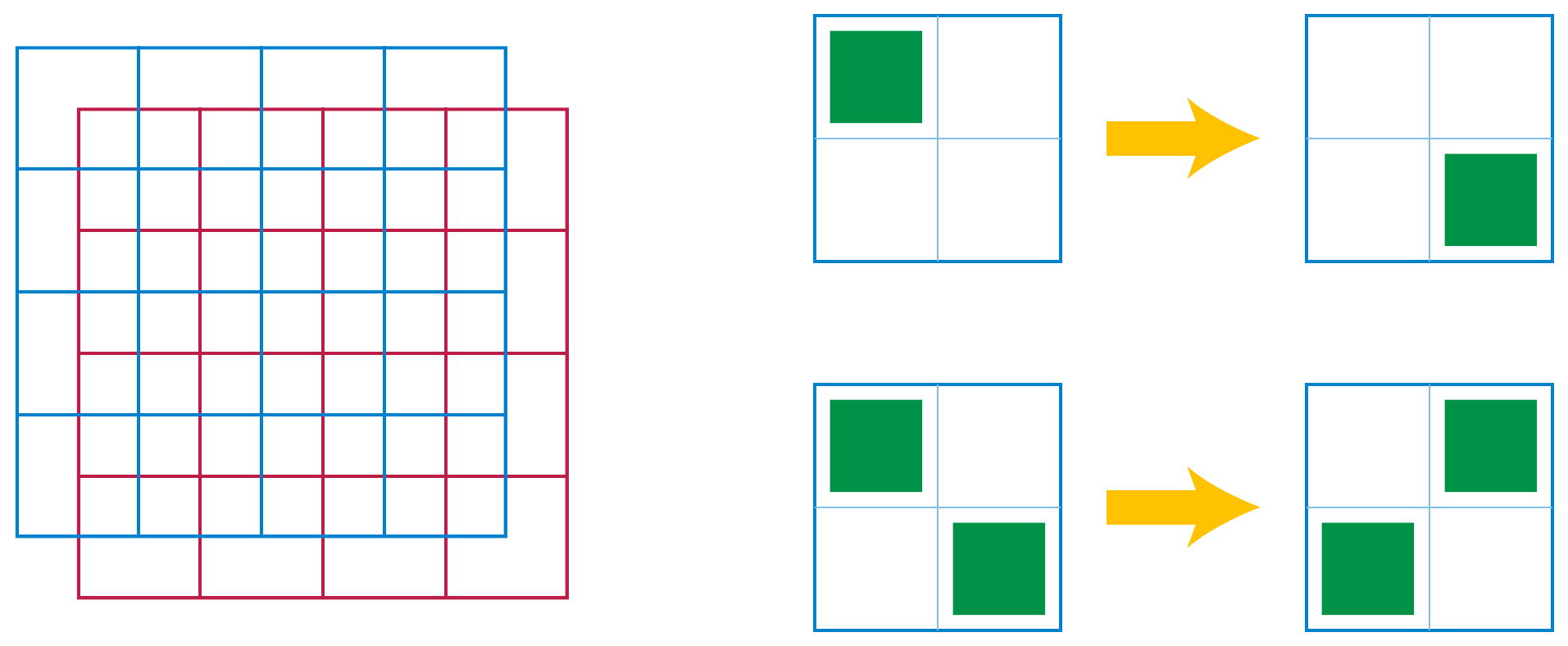}
\caption{The billiard-ball model. Left: the Margolus neighborhood breaks up the square grid of cells into $2\times 2$ square blocks in two alternating ways, as shown by the blue and red blocks. Right: Blocks with one live cell, or with two diagonal live cells, change in the ways shown; all other blocks remain unchanged.}
\label{fig:bbm}
\end{figure}

We will consider in more depth the \emph{billiard-ball model} or BBM block-cellular automaton, devised by Margolus to simulate reversible logic, universal Turing machines, and other reversible cellular automata~\cite{Mar-PD-84}. This cellular automaton uses the \emph{Margolus neighborhood}, in which the cells of a square grid are grouped into square blocks of four cells, in two alternating ways, with the corners of the square blocks in even generations of the automaton forming the centers of the square blocks in odd generations (\cref{fig:bbm}, left). Cells have two states, dead or alive. The transition function of the automaton acts independently within each square block. In a block with a single live cell, the updated state has again a single live cell in the diagonally opposite position. In a block with two diagonally placed live cells, all four cells change from live to dead or vice versa. All other blocks remain unchanged (\cref{fig:bbm}, right).

Patterns within this automaton can simulate any Fredkin-gate circuit. Wires of the circuit are simulated by diagonal paths along which travel signals consisting of pairs of live cells. Each Fredkin gate is simulated by fixed blocks of live cells that interact with these signals, causing them to change their paths in ways that match the behavior of the gate. Other fixed blocks of live cells can be used to bend or delay wires so that signals meet up in synchrony or cross without interacting. These issues of reversible circuit layout and synchronization in the BBM cellular automaton have been discussed in detail in multiple previous works, to which we refer the reader for details~\cite{Mar-PD-84,Hay-SA-84,Dur-CBC-02,Mor-TCS-08}. As well as reversible logic circuits, this automaton can simulate arbitrary Turing machines and other reversible cellular automata~\cite{Dur-CBC-02}. However, this work leaves open the question: what is the power of this system when the grid is not infinite, but of bounded size, with periodic boundary conditions?

\begin{theorem}
\label{thm:bbm-hardness}
Simulating a given number of iterations of a BBM pattern of bounded size, given as an array of initial cell states with periodic boundary conditions, is complete for $\ib$.
\end{theorem}

\begin{proof}
The problem can clearly be expressed as an iterated bijection, as each step of the simulation is a polynomial-time invertible function. Therefore it is in $\ib$.

By \cref{obs:circuit-completeness}, simulating the behavior of a reversible logic circuit, with its outputs fed back into its inputs, for a given number of steps, is complete for $\ib$. We outline a many-one reduction from this problem to the simulation of BBM patterns, using previously-described ways of simulating circuit components in BBM. The reduction lays out the given circuit as a BBM pattern, including wire-bending and delay circuits that feed the output signals from the circuit back into the inputs, delayed so that the outputs all return to the inputs in synchrony. We use known methods for laying out circuits or other bounded-degree planar graphs onto grid graphs of polynomial area~\cite{Tam-SICOMP-87}, with the layout oriented diagonally with respect to the BBM cell grid, in accordance with the diagonal movement of the circuit signals. The circuit requires only a bounding box of size polynomial in the circuit size, requires only a polynomial amount of delay for re-synchronization (and a corresponding amount of area for the delay circuits), and performs all its simulations of the given reversible logic gates within a polynomial number of steps.

Hundreds of published NP-completeness proofs already follow this same approach of using orthogonal layouts of circuits (often, of circuits for 3-satisfiability problems) in their reductions; for a typical example, see \cite{MooRob-DCG-01}. The use of delay gadgets to correctly synchronize or desynchronize signals within the billiard-ball model is also standard~\cite{Dur-CBC-02}. Therefore, we omit the details of these constructions.

The output of the given circuit after a given number of iterations can be obtained by simulating the translated BBM pattern, with input signals set to match the inputs to the given circuit, for a number of steps equal to the product of the number of iterations for the circuit and the time for an input signal to return to the same point in the BBM pattern.
\end{proof}

\subsection{Other known universal reversible cellular automata}
The completeness of simulating other universal reversible cellular automaton rules would need to be considered case-by-case, depending on how the universality of those other rules has been proved. For instance, Toffoli~\cite{Tof-JCSS-77} transforms arbitrary non-reversible cellular automata of dimension $d$ into reversible automata of dimension $d+1$ by making the higher-dimensional automaton construct the entire time-space diagram of the lower-dimensional automaton. However, this also has the effect of increasing the space (number of cells) required for the higher-dimensional automaton to accurately perform this simulation, to be proportional to the product of space and number of simulated steps of the lower-dimensional automaton. Because the space bound for Toffoli's method is not polynomial in the space of the simulated automaton, this method cannot be formulated as a polynomial-time many-one or Turing reduction from one problem to another, and cannot be used for proving $\ib$-completeness. Similarly, Morita~\cite{Mor-JCA-07} has shown how to simulate cyclic tag systems by a finite pattern in a universal one-dimensional reversible cellular automaton, but the correct behavior of this automaton requires a number of cells proportional to the number of steps of the automaton, so that garbage states from the automaton do not wrap around into the part of the automaton used for describing the rules of the tag system. Again, this need for a number of cells that depends in some way on the time complexity of the simulated computation prevents this method from being used to prove $\ib$-completeness.

\subsection{Dimension reduction}

The $\ib$-completeness of the two-dimensional BBM automaton, and our failure to translate the existing universality proofs of one-dimensional reversible cellular automata into $\ib$-completeness, raise a natural question: can simulating a one-dimensional reversible cellular automaton be $\ib$-complete? We answer this question affirmatively, by providing a one-dimensional simulation of any two-dimensional Margolus-neighborhood reversible cellular automaton, using the following ingredients:

\begin{figure}[t]
\centering\includegraphics[width=0.6\textwidth]{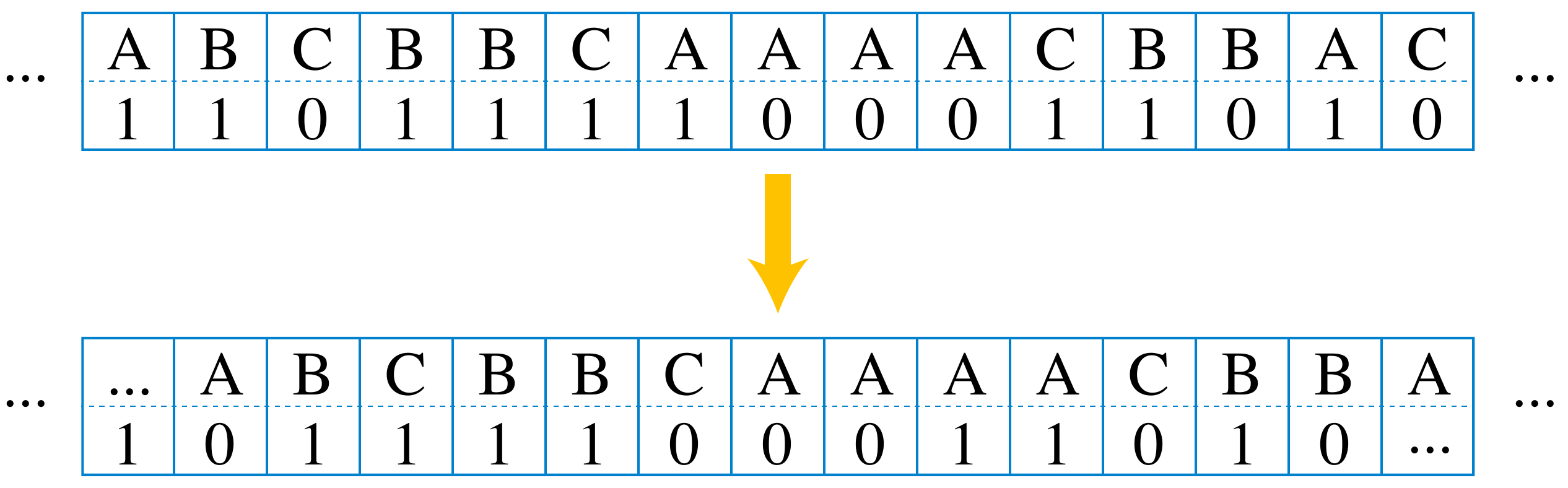}
\caption{A two-track automaton in which the update rule for the top track copies the left neighbor and the update rule for the bottom track copies the right neighbor.}
\label{fig:band}
\end{figure}

\paragraph{Tracks.} It will be convenient to think of each cell of a one-dimensional cellular automaton as being composed of multiple \emph{tracks}, each containing a finite state, with possibly different sets of states for different tracks. It is possible for each track to have an update rule that is independent of the states in other tracks, with the value of a single track in a cell computed as a combination of the values in the same track of neighboring cells.  \cref{fig:band} shows an example with two tracks, in which the update rule for the top track copies the left neighbor while the update rule for the bottom track copies the right neighbor, causing the states of the tracks to move relative to each other while remaining otherwise unchanged. (This example is from a family of reversible automata described by Boykett~\cite{Boy-TCS-04} as having an update rule that acts on the values of whole cells by combining the left and right neighbors using an algebraic structure called a rectangular band.) Alternatively,  the state of one track can control the update rule of another track. As long as these controlled update rules remain individually reversible, the whole automaton will again be reversible, with a reverse dynamics that computes the predecessor value of the controlling track and then uses it to control the update rule of the other track.

 The number of states of the whole cell is then the product of the numbers of states within each track. These numbers grow quickly, so the automata resulting from multi-track constructions will in general have many states. However, if we have a fixed number of tracks with a fixed number of states in each track, the total number of states remains finite, as is required for a cellular automaton.

 \paragraph{Partitioning automata.}
 
A general construction for one-dimensional reversible automata of Imai and Morita~\cite{ImaMor-TCS-96}
can be thought of as having three tracks per cell: a left track, center track, and right track.  The left and right tracks must have equal sets of available states; the states of the center track can differ. The update rule for the automaton performs two operations (as a single automaton step):
 \begin{itemize}
 \item Swap the value in the right track of each cell with the value in the left track of its right neighbor.
 \item Apply a bijective transformation to the state of each cell (the combination of the states of all three of its tracks), independently of the states of its neighbors.
 \end{itemize}
 All one-dimensional cellular automata defined in this way are automatically reversible. The reverse dynamics can be described similarly, as applying the inverse bijective transformation and then swapping values in the same way. Although this is also a cellular automaton, it is not a partitioning automaton of the same type, because the bijection and the swap are performed in a different order.

\paragraph{Firing squad synchronization.}

We use a reversible solution to the firing squad synchronization problem found by Imai and Morita~\cite{ImaMor-TCS-96}.
 
 \begin{lemma}[Imai and Morita~\cite{ImaMor-TCS-96}]
 \label{lem:firing-squad}
There is a one-dimensional reversible cellular automaton with the following behavior. First, its states can be partitioned into three sets: quiescent, active, and firing. Second, a cell that is quiescent remains quiescent throughout the evolution of the automaton; therefore, the behavior of any pattern can be described purely by considering its contiguous subsequences of non-quiescent cells. Third, for every $k$ there exists a pattern $P_k$, consisting of $k$ cells in active states, bounded on both sides by quiescent cells, with the following property: for all $0\le i<3k$, the pattern resulting from $P_k$ after $i$ steps consists only of active states, but the pattern resulting from $P_k$ after exactly $3k$ steps consists only of firing states.
 \end{lemma}

Some additional detail on how this firing squad computation works will be important. It is a partitioning automaton, where the center track holds the state of each cell: quiescent, active, and firing, with additional information about several more specific types of active cells. On quiescent cells, the bijective transformation of the partitioning automaton is the identity. The left and right tracks hold ``signals'' that move leftwards or rightwards through the automaton, interacting with other cells as they do. The leftmost cell of the initial pattern $P_k$ is an active cell in a ``general'' state, with the remaining active states being ``soldiers''. The general sends out two signals to its left, and transitions to a ``waiting'' state, waiting for a signal to return from the left. The faster of the two signals sent out in the same direction by the general bounces off the boundary of the pattern, and meets the slower signal in the middle of the pattern. When the signals meet, their interaction produces two new generals in central cells, one sending signals to the left and the other to the right, which again transition to a waiting state, waiting for a signal to return from the direction it was sent. In this way, the pattern is split recursively into two subpatterns which behave in the same way, recursively splitting into smaller subpatterns, until at the base level of the recursion all of the constant-length patterns fire simultaneously.

An important consideration, for our purposes, is the behavior of the ``waiting'' states. In a cell waiting for a signal from the left, the bijective transformation acts as the identity on the right track, and this track otherwise does not affect the behavior of the cell. The cells waiting for signals from the right are symmetric. Effectively, these states partition the pattern into parts that do not interact with each other, without the need for quiescent states. If the initial pattern of Imai and Morita is run for $3k/2+O(1)$ steps, it reaches a configuration in which the leftmost active cell is waiting on a signal from its right, and a central active cell is waiting on a signal from its left. The pattern of length $k/2$ between these two waiting cells will then fire in $3k/2-O(1)$ more steps, regardless of any modification to the states in any other part of the pattern, because these two waiting cells block all interaction from other parts of the pattern. We can eliminate the division by two in these formulas by doubling the size of the initial pattern, and formalize this as the following observation:
 
 \begin{observation}
 \label{obs:noq}
 For the Imai--Morita firing squad automaton, for all $k$, there exist patterns of $k$ consecutive active cells that will remain active for $3k-O(1)$ steps and then simultaneously fire, regardless of how the cells outside this pattern are initialized. Additionally, these patterns can be constructed in time polynomial in~$k$.
\end{observation}
  
\paragraph{Strobed synchronization.} We will need to use update rules that, every $k$th step for a variable numerical parameter $k$, perform a different step than the usual computation in the other steps. This can be thought of by analogy to a strobe light, which provides brief flashes of one condition (bright light) interspersed with longer periods of a different condition (darkness).  This is not something that can be directly defined into the behavior of a cellular automaton, because directly storing the number of the current step modulo~$k$ would use a number of bits of information that is logarithmic in~$k$, rather than being encodable into a finite state. Instead, we will simulate this behavior by using tracks that perform the firing squad synchronization computation of Imai and Morita, repeating spatially. We say that a state of one-dimensional cellular automaton is \emph{spatially repeating} with pattern $P$ and period $k$ if $P$ is a sequence of automaton states of length $k$ and the state is formed by concatenating an infinite sequence of copies of $P$. (Such a state will also be repeating for any period that is a multiple of $k$.) The states of such an automaton continue repeating with the same period for all subsequent time steps. They have the same behavior as an automaton run with the same rules on a finite cycle of cells of length $k$ containing a single copy of $P$. (Connecting the start and end of $P$ in this way to form a cycle of cells is commonly referred to as using \emph{periodic boundary conditions}.) 
   
 \begin{lemma}
 \label{lem:strobe-automaton}
 There is a one-dimensional reversible cellular automaton with the following behavior. Its states can be partitioned into two sets: active, and firing. For every $k$ and $t$ with $1<t\le3k-O(1)$ there exists a pattern $P_{k,t}$, consisting of $k$ cells in active states, with the following property: for all $i$, the spatially repeating state with pattern $P_{k,t}$ and period $k$ consists only of firing states at time steps $i$ where $i$ is a multiple of~$t$, and consists only of active states at other time steps.
 \end{lemma}

\begin{proof}
We use an automaton with six tracks (top left, top center, top right, bottom left, bottom center, and bottom right). We use a variation of partitioning automaton dynamics in which, in each step, we perform the following three steps:
\begin{enumerate}
\item Each cell exchanges the value in its top right track with the value in the top left track of its right neighbor, as in the forward partitioning automaton dynamics.
\item We apply a bijective transformation to the combination of the six values in the six tracks, independently for each cell.
\item Each cell exchanges the value in its bottom right track with the value in the bottom left track of its right neighbor, as in the reverse partitioning automaton dynamics.
\end{enumerate}
As in a standard partitioning automaton, the resulting automaton will automatically be reversible.
We define the available states for these tracks to be the same as for the left, right, and center tracks of the firing-squad automaton of \cref{lem:firing-squad}, except that we do not use quiescent states. Let $F$ be the bijective transformation used to define this firing-squad rule.  We use the following rules to define the transformations of our six-track automaton.
\begin{itemize}
\item We call a state that has a firing top center cell and an active bottom center cell ``top-lit'', and we call a state that has a firing bottom center cell and an active top center cell ``bottom-lit''. We define the bijective transformation on a top-lit cell to swap the values in its top and bottom tracks, producing a bottom-lit cell.
\item For a cell that is not top-lit, and for which the result of the transformation would not be bottom-lit, we define the bijective transformation to apply $F$ to the three top tracks and $F^{-1}$ to the three bottom tracks.
\item The definitions above leave undefined the successors of cells for which the $(F,F^{-1})$ transformation would be bottom-lit. They also leave undefined an equal number of predecessors of cells, the ones that would be reached by an $(F,F^{-1})$ transformation on a top-lit cell. Pair these missing successors and predecessors arbitrarily to form a bijective transformation. (These transitions will not be used by the patterns we construct, so the behavior of the automaton for cells of these types is unimportant in achieving the desired behavior, but it still needs to be a bijection in order to define a reversible automaton.)
\end{itemize}
Given these rules, initialize the top tracks of the automaton with the pattern of \cref{obs:noq}, for parameter $k$, in the all-cells-firing configuration obtained after $3k-O(1)$ steps of the firing-squad automaton. Initialize the bottom tracks of the automaton with the configuration $t-1$ steps earlier.

For this pattern, all cells begin top-lit, and the automaton will immediately swap the top and bottom tracks of all cells. Then, for the next $t-1$ steps, it will follow the forward evolution of the Imai--Morita firing squad automaton on the top tracks, and the reverse evolution on the bottom tracks, reaching the same configuration that it started with and repeating the same behavior. These steps cannot reach one of the arbitrary transitions in the final bullet of our transition rule definition, because that could only happen when the result of an $(F,F^{-1})$ transformation would be bottom-lit, and by construction the $t-1$ predecessors of the all-firing bottom-track state are not themselves firing.
\end{proof}
  
 \begin{lemma}
 \label{lem:strobe}
 Let $R_1$ and $R_2$ be two different update rules for reversible cellular automata, both operating on the same finite set $S$ of states. Then there is a one-dimensional reversible cellular automaton with $O(|S|)$ states that can simulate the following behavior, for every positive integer $t$: given an array of cells with states in $S$, update the cells in the array using rule $R_1$ on every $t$th step, and using rule $R_2$ on all other steps. Further, the automaton that simulates this behavior can be made to be spatially repeating (necessarily causing the simulated behavior to also be spatially repeating) for all periods larger than $t/3+O(1)$ and all patterns with the given period.
 \end{lemma} 
 
 \begin{proof}
 We simulate this behavior using a multi-track cellular automaton, where one track $T_1$ contains the states from $S$ in the simulated array of cells, and the remaining tracks perform the operations of \cref{lem:strobe-automaton} regardless of the values in track $T_1$.  We will apply rule $R_1$ to a cell of track $T_1$ whenever the state of the remaining tracks is top-firing, and rule $R_2$ otherwise. In the reverse dynamics, rule $R_1^{-1}$ should be applied whenever the state of the remaining tracks is bottom-firing, and otherwise $R_2^{-1}$ should be applied.
 
We initialize the tracks for the automaton of \cref{lem:strobe-automaton} according to that lemma, with the same parameter~$t$, spatially repeating with a suitable period $k$. If $S$ is spatially repeating then $k$ should be chosen to be the same period; otherwise, $k$ can be chosen arbitrarily to meet the conditions of \cref{lem:strobe-automaton}.
\end{proof}
 
 \paragraph{Helical boundary conditions.} The hard-to-simulate patterns in the BBM automaton remain confined to within their bounding box, rather than expanding beyond it. This confinement means that, for a grid of cells with a sufficient constant-sized margin beyond the bounding box, the boundary conditions of the grid are irrelevant: any reasonable choice of boundary conditions will lead to the same evolution. Therefore, we are free to choose boundary conditions that are easy to transform into a one-dimensional automaton, rather than having some particular choice forced on us.
 
 \begin{figure}[t]
 \centering\includegraphics[width=0.4\textwidth]{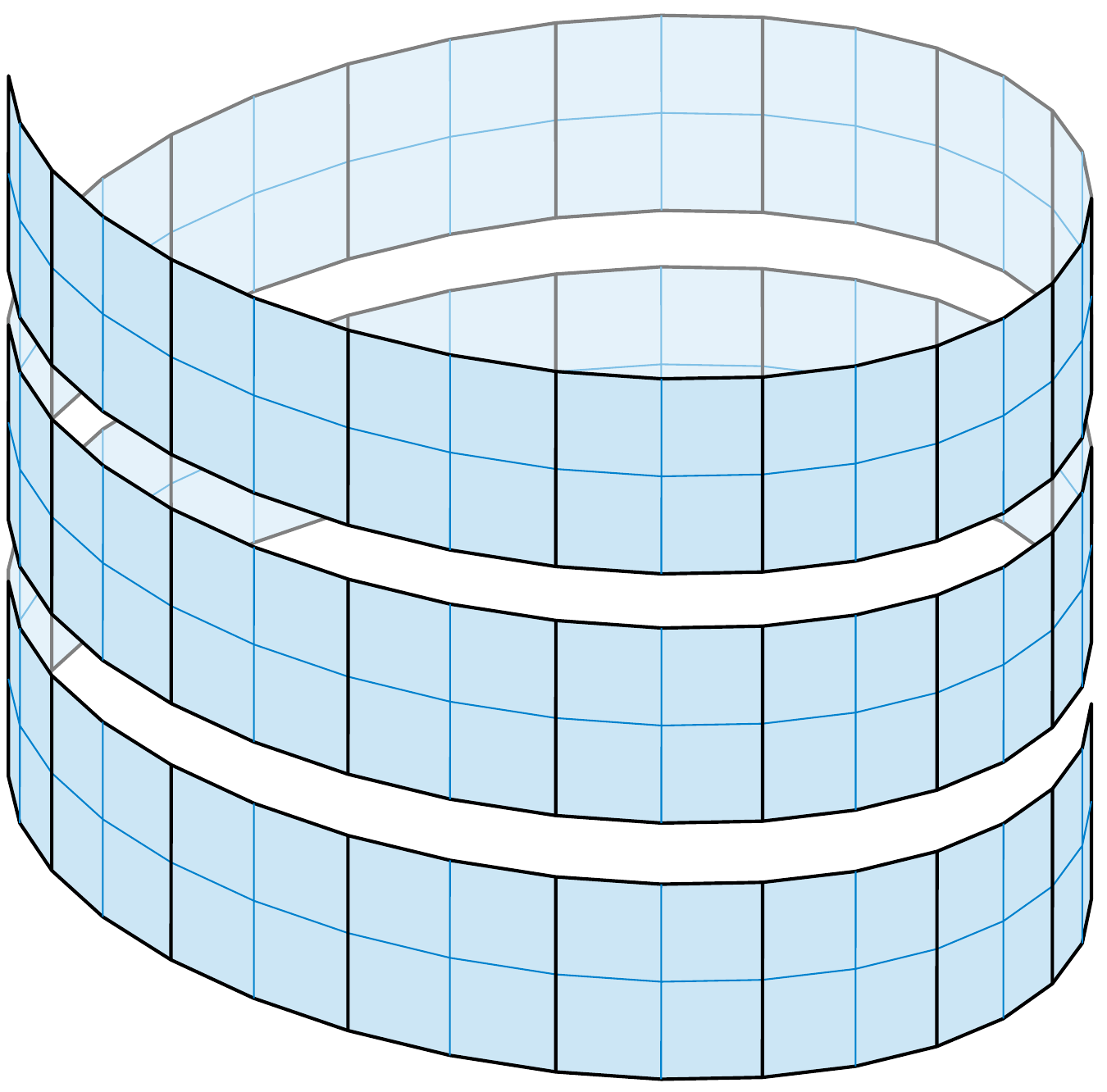}
 \caption{Helical boundary conditions for the two-dimensional Margolus neighborhood (shown here with circumference 32 in an exploded view with spacing between rows of squares) transform its behavior for a single time step into that of a two-track one-dimensional cellular automaton.}
 \label{fig:helical}
 \end{figure}
 
In a single update of the Margolus neighborhood used by the BBM automaton, the cells are partitioned into four-cell squares, and only interact with each other within these squares. Consider an arrangement of these four-cell squares into an infinite horizontal strip. We define \emph{helical boundary conditions} for this strip, with \emph{circumference} $c$, by adding vertical connections from each of these four-cell squares upwards to the square offset from it by $c$ units leftward along the strip, and downward to the square offset from it by $c$ units rightward, as depicted in $\cref{fig:helical}$. In order to get the squares to line up, we require that $c$ be even. This creates a pattern of cell connectivity that locally (within regions of width less than $c$) is indistinguishable from the infinite square grid, although globally it has the topology of a cylinder, not the same as a grid. If we map a system of cells, connected in this way, onto a two-track one-dimensional automaton, in which each cell of the one-dimensional automaton holds two cells of the Margolus neighborhood, it will be easy for the one-dimensional automaton to simulate the updates in the Margolus neighborhood that use four-cell squares in odd-numbered steps, aligned with the given strip. However, the updates in even-numbered steps combine information from cells $c$ units apart from each other in the strip, and it is not obvious how to perform those updates using a one-dimensional automaton whose neighborhood size does not depend on~$c$. Our eventual solution to this problem will use strobing synchronization to permute the cell states into a position where interacting cells are again adjacent within the one-dimensional strip.

We define \emph{toroidal boundary conditions} of circumference $c$ and period $p$ (requiring $p$ to be even and greater than $c$) by using both periodic boundary conditions of period $p$ for the one-dimensional strip of four-cell squares, and helical boundary conditions to define vertical neighbors of each square. The resulting system of cells is again locally (within regions of width less than $c$ and height less than $p/c$) indistinguishable from the infinite square grid, although globally it has the topology of a torus.

With the pieces we need all defined, we are now ready to describe our one-dimensional simulation of two-dimensional Margolus-neighborhood reversible automata.

\begin{theorem}
\label{thm:dim-redux}
Every reversible cellular automaton with the two-dimensional Margolus neighborhood and with $s$ states per cell, running on a system of cells with helical boundary conditions with any even circumference $c$, can be simulated by a one-dimensional cellular automaton with $O(s^2)$ states per cell, with a system of states and an update rule that does not depend on $c$, and with $c/2+1$ steps of the one-dimensional automaton for every simulated step of the two-dimensional automaton. The simulated automaton can be made to have toroidal boundary conditions for any period larger than the circumference, giving the one-dimensional automaton periodic boundary conditions with the same period.
\end{theorem}

\begin{proof}
Let $t=c/2+1$. We simulate the two-dimensional automaton using a multi-track one-dimensional automaton two of whose tracks represent the upper and lower rows of cells in the helical boundary conditions, and we use strobing synchronization with strobe period $t$, using more tracks. We use one more track to store a single bit of information for each one-dimensional cell, indicating whether its first two tracks should be combined with the cell to the left or with the cell to the right to form the four-cell squares of the Margolus neighborhood; we set the initial state of these bits in strict alternation between consecutive cells.

As in \cref{lem:strobe}, all states of the strobing track will be top-lit in one step, followed by $t-1$ steps in which they are not top-lit, in a temporally-repeating pattern. When a cell is top-lit, we perform the update in the Margolus neighborhood given by the reversible dynamics of the given two-dimensional automaton, with the following small modification: we swap the resulting cell values between the top and bottom tracks. As a result of this step, the cell values of the two-dimensional automaton are all computed correctly, but are placed in cells that are not adjacent to their neighbors in the next update. The cell values that are now in the top track of the one-dimensional simulation need to be paired with values that are now in the bottom track but are $c$ units farther to the right. To fix this incorrect placement, in each of the $t-1$ subsequent steps of the one-dimensional automaton, we slide the top track rightward one step and the bottom track leftward one step, according to the rectangular band dynamics depicted in \cref{fig:band}. After these sliding movements of all the cell states, they will once again be placed in a position where the one-dimensional automaton can perform a Margolus-neighborhood update.

Between one top-lit step and the next, the values that were in two vertically-adjacent squares (according to the adjacency pattern of the helical boundary conditions, although far from each other in the one-dimensional simulation) are shifted halfway around the helix in opposite directions, landing on different tracks of a single one-dimensional cell. In this same span of steps, we need to update the bit of information on the final track indicating whether each cell should look left or right to form a Margolus neighborhood in the next step. This update should be done in a way that matches the correct alignment of these squares, halfway around the helix, in alternating steps of the two-dimensional automaton. The required update to the final track depends on the parity of the number of Margolus-neighborhood squares in a single cycle around the helix, $c/2=t-1$. When there are evenly many squares in this cycle (true when $t$ is odd), the final-track states should all be flipped from one light step to the next; otherwise, when $t$ is even, they should all be left unchanged. This can most easily be accomplished by flipping these states at every step, regardless of the state of the strobing track.
\end{proof}

It seems likely that the divisibility condition on the circumference of the helical boundary conditions in \cref{thm:dim-redux} can be relaxed using a more general strobing synchronization mechanism, but we do not need this added generality for the following result:

\begin{theorem}
\label{thm:1d-rca-completeness}
There is a one-dimensional reversible cellular automaton for which simulating any given number of iterations, with periodic boundary conditions, is complete for $\ib$.
\end{theorem}

\begin{proof}
Apply \cref{thm:dim-redux} to produce a one-dimensional simulation of the BBM two-dimen\-sion\-al automaton, for a hard instance of BBM generated by \cref{thm:bbm-hardness}, and for toroidal boundary conditions with circumference and period both large enough to make no difference to the dynamics of BBM within the bounding box of live cells of the instance.
\end{proof}

The number of states in the automaton resulting from this construction is $2^3\cdot 90^2=64800$, a finite but large number. (The two strobing tracks have 90 states rather than the 99 states of Imai and Morita because they do not use quiescent states.) It would be of interest to find a $\ib$-complete one-dimensional reversible cellular automaton with significantly fewer states.

\section{Piecewise linear bijections}

The study of iterated behavior of piecewise linear maps and of bijective maps are both central to the theory of dynamical systems. Well-known mappings in this area that combine both characteristics include Arnold's cat map
\[(x,y)\mapsto (2x+y,x+y)\bmod 1\]
and the baker's map
\[(x,y)\mapsto \left(2x\bmod 1, \frac{y+\lfloor 2x\rfloor}{2}\right),\]
both on the unit square~\cite{Orn-AMM-71,DysFal-AMM-92}.

\begin{figure}[t]
\centering\includegraphics[scale=0.4]{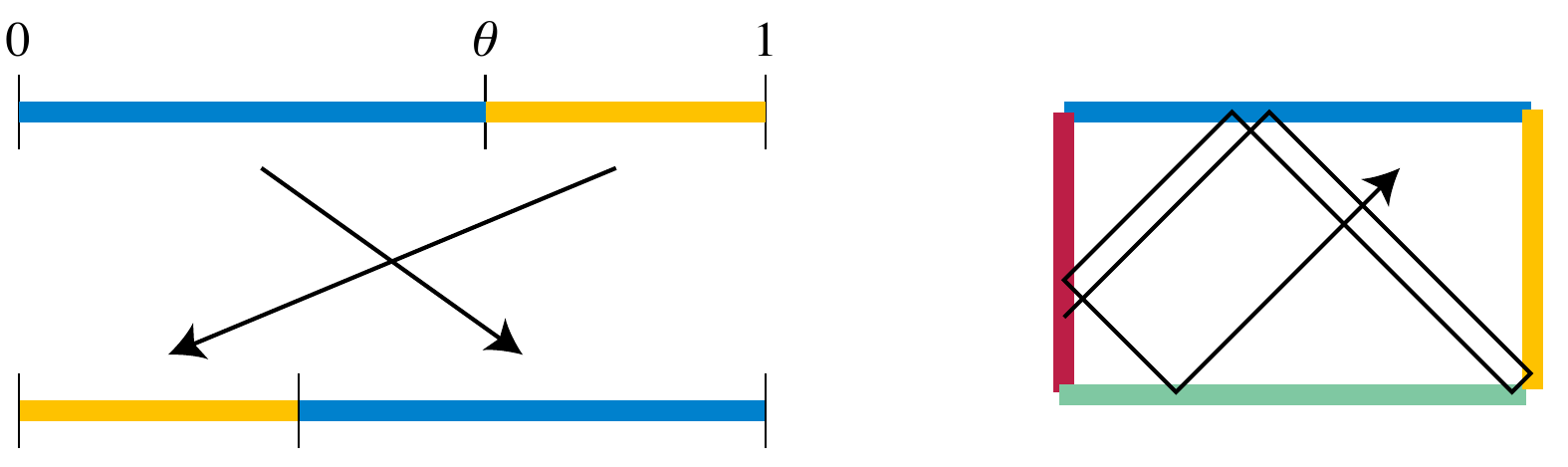}
\caption{Left: The interval exchange transformation $x\mapsto (x+\theta)\bmod 1$. Right: More complicated interval exchange transformations can be used to model reflections in mirrored polygons.}
\label{fig:ix-reflex}
\end{figure}

A prominent family of one-dimensional systems are the \emph{interval exchange transformations}~\cite{Kea-MZ-75}. These are piecewise linear bijections that partition a half-open interval into subintervals, permute the subintervals, and translate each subinterval into its permuted position (\cref{fig:ix-reflex}). The computational complexity of iterated interval exchange transformations, and their application in modeling light reflections within mirrored polygons, was an initial motivation for this paper. Even the most simple nontrivial interval exchange, the transformation $x\mapsto (x+\theta)\bmod 1$, has interesting iterated behavior, including Steinhaus's three-gap theorem according to which there are at most three distinct intervals between consecutive values in the sorted sequence of the first $n$ iterates~\cite{Shi-AMM-18}.

For the purposes of computational complexity it is more convenient to consider mappings that act on discrete sets rather than on continuous spaces like the entire unit square. In this light, it is common, for instance to study the effect of Arnold's cat map on grid points, such as the positions of a discrete set of pixels~\cite{DysFal-AMM-92}; indeed, its name comes from an example given by Arnold of a picture of a cat being transformed in this way. It is important to note, however, that restricting the domain of a function in this way can change whether it is bijective. Arnold's cat map is bijective on square grids, for instance, but the baker's map is not: each step halves the vertical separation of the grid.

\begin{figure}[t]
\includegraphics[width=\textwidth]{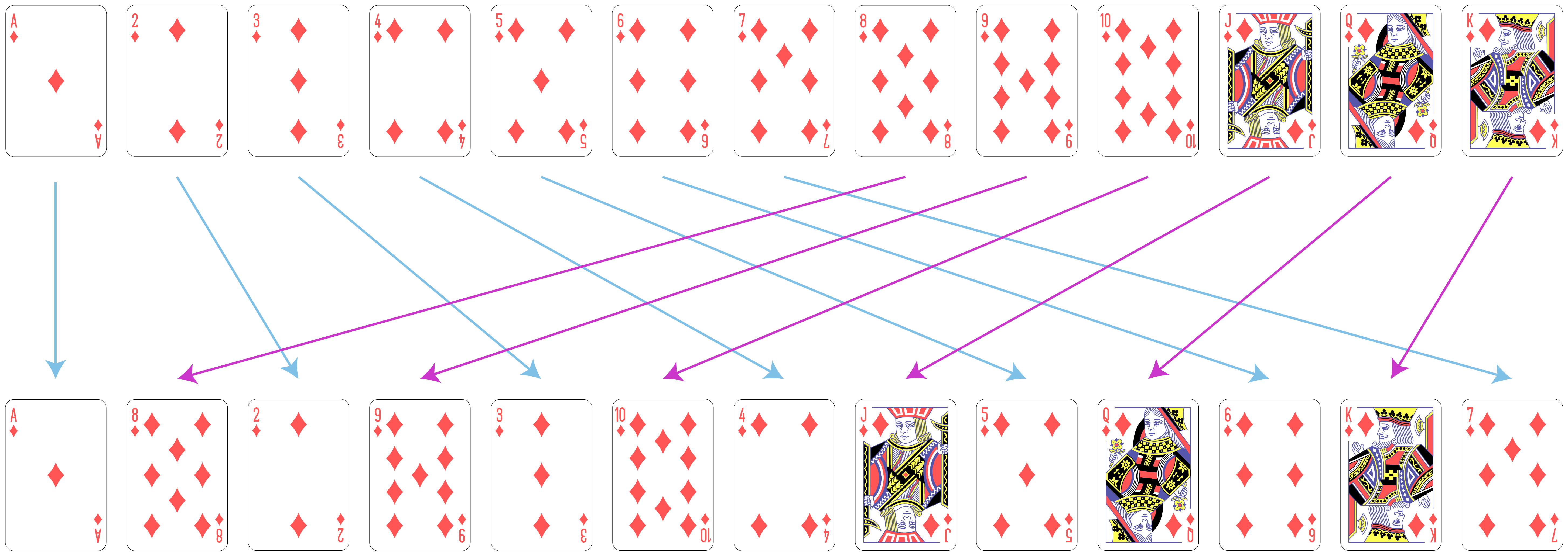}
\caption{A perfect riffle shuffle of $n$ cards can be represented as a piecewise linear transformation with two pieces acting on the first $\lceil n/2\rceil$ cards (blue arrows) and the remaining $\lfloor n/2\rfloor$ cards (magenta arrows). Card images from public domain file ``English pattern playing cards deck.svg'' by Dmitry Fomin on Wikimedia commons.}
\label{fig:shuffle}
\end{figure}

As an example in the other direction, of a piecewise linear transformation that is bijective on integers but not on continuous intervals, consider the familiar perfect riffle shuffle of a deck of $n$ cards~\cite{DiaGraKan-AAM-83}, which (if the cards are represented by integers in the range from $0$ to $n-1$) can be expressed as the piecewise linear transformation
\[
i\mapsto
\begin{cases}
2i & \mbox{if } i<n/2\\
2i-n & \mbox{if } i>n/2\mbox{ and } n \mbox{ is odd}\\
2i-n+1 & \mbox{if } i>n/2\mbox{ and } n \mbox{ is even.}\\
\end{cases}
\]
See \cref{fig:shuffle} for an example with $n=13$. We will use these shuffle transformations as components in a hardness proof for a problem of computing iterated piecewise linear bijections.

\subsection{The piecewise linear bijection problem}

\begin{definition}
We formulate the \emph{iterated piecewise linear bijection problem} as follows: the input is a triple $(x,n,T)$, where $x$ and $n$ are integers, and $T$ is a piecewise linear bijection on a range of integers that includes $x$, described by specifying integer endpoints of each piece of the bijection, together with two integer coefficients for the linear transformation of that piece. The output is $T^{(n)}(x)$.
\end{definition}

\begin{observation}
\label{obs:plb-poly}
Given an input that is in the syntactical form of an input to the iterated piecewise linear bijection problem, we can test in polynomial time whether it correctly describes a piecewise linear bijection. If it does, we can implement a single iteration of the bijection, or of the inverse of the bijection, in polynomial time.
\end{observation}

\begin{proof}
To test whether an input of this form describes a piecewise linear bijection, sort the endpoints of the specified pieces to check that they form disjoint intervals whose union is a single interval, check that each of the specified transformations maps each piece into this union, and check that no two specified transformations have intersecting images.

The images of any two linear pieces of the given input lie in two arithmetic progressions, and testing whether they intersect can be done using greatest common divisors to form the intersection of these progressions~\cite{Ore-AMM-52}. In more detail, suppose that the image of one piece lies within an interval of the progression of values that are $a$ mod $n$, and that the image of a second piece lies within an interval of a progression of values that are $b$ mod $n$. Let $g=\gcd(m,n)$; then the intersection of the two progressions is either empty (if $a\ne b$ mod $g$) or has period $mn/g$ (otherwise). When it is non-empty, the extended Euclidean algorithm can be used to find numbers $m'$ and $n'$ with $mm'+nn'=g$, and the intersection of the progressions consists of the values congruent to $(ann' + bmm')/g$ mod $mn/g$. To test whether the two images intersect, we need only compute the coefficients of this progression and test whether it has any values in the interval between the upper and lower ends of both images.

To implement a single iteration of the bijection, on input $x$, find the piece containing $x$ and apply its transformation. 
To implement the inverse of the bijection, find the piece whose image contains $x$ and apply the inverse of its linear transformation.
\end{proof}

Although the inverse of a piecewise linear bijection is linear on each image of a piece, this may not describe the inverse as a piecewise linear bijection, because the images might not be intervals. Describing the inverse of a piecewise linear bijection as another piecewise linear bijection could produce significantly more pieces. To allow for more general descriptions of transformations, it will be convenient for us to consider compositions of piecewise linear bijections, repeated in a fixed sequence. As the following lemma shows, this can be done by combining the sequence into a single piecewise linear bijection on a larger range, without a significant increase in complexity.

\begin{lemma}
\label{lem:compose-plb}
Let $T$ be the composition of a sequence of $k$ piecewise linear bijections $T_1, T_2,\dots T_k$ on the integers in the range $[0,n)$. Then there exists a single piecewise linear bijection $\widehat T$ on the range $[0,kn)$, such that for all $x$ in $[0,n)$, $T(x)=\widehat T^{(k)}(x)$. The number of pieces needed to define $\widehat T$ as a piecewise linear transformation is the sum of the numbers of pieces in each $T_i$.
\end{lemma}

\begin{proof}
For each linear transformation in each of the piecewise linear bijections $T_i$, mapping a subinterval $[a,b)$ to another subinterval $[c,d)$, make a corresponding transformation in $\widehat T$ that maps $[a+(i-1)n,b+(i-1)n)$ to $[c+in,d+in)$.
For the final iteration $T_k$, the image of this transformation should be taken modulo $kn$, so that it wraps around to $[0,n)$. The transformation $\widehat T$ defined by combining these pieces of transformations has the property that, when it is iterated $k$ times on a starting value $x$ in the range $[0,n)$, the $i$th iteration maps $x$ into the range $[in,(i+1)n)$, and that the behavior of this iteration modulo~$n$ is the same as $T_i$. The last of the $k$ iterations maps $x$ back into $[0,n)$ in the same way. Therefore, for $x\in[0,n)$, $T(x)=\widehat T^{(k)}(x)$.
\end{proof}

\subsection{Permuting the bits of a binary number}

Binary rotation or circular shift operates on numbers in the range $[0,2^k)$ as follows. Represent any number as a binary string, with the most significant bit on the left and least significant bit on the right. A left circular shift by $i$ units, for $1\le i<k$, moves each bit value into the position~$i$ steps to the left, with the most significant $i$ bits wrapping around into the least significant $i$ positions. A right circular shift by $i$ units performs the opposite transformation, moving each bit value $i$ steps to the right, with the least significant $i$ bits wrapping around into the most significant $k$ positions. A right circular shift by $i$ units is the same as a left circular shift by $k-i$ units.

\begin{observation}
\label{obs:rot}
A left circular shift by one unit on the range $[0,2^k)$ can be expressed as a piecewise linear bijection with two pieces.
\end{observation}

\begin{proof}
It is the transformation
\[
x\mapsto \begin{cases}
2x & \mbox{if } x<2^{k-1}\\
2x-2^k+1 &\mbox{otherwise.}\\
\end{cases}
\]
\end{proof}

This is just the riffle shuffle example described earlier, in the case where the number of values being shuffled is a power of two. By applying \cref{lem:compose-plb} we can compose multiple one-unit shifts to obtain circular shifts of larger numbers of units. We can also compose these shifts in more complex ways to obtain other bit permutations:

\begin{lemma}
\label{lem:plb-permute}
Let $C$ be a subset of $[0,k)$, interpreted as bit positions in the $k$-bit binary values, with $0$ as the least significant (rightmost) position, and $k-1$ as the most significant (leftmost) position. Then there exists a function $f$ on $k$-bit binary values that permutes the bits so that the positions in $C$ are moved to the $|C|$ most significant bits, such that both $f$ and $f^{-1}$ can be expressed as compositions of $O(|C|k)$ piecewise linear bijections with $O(1)$ pieces each.
\end{lemma}

\begin{proof}
We express $f$ as a composition of piecewise linear bijections using induction on $|C|$. As base cases, if $|C|=0$, we may let $f=f^{-1}$ be the identity function, expressed as a piecewise linear bijection with one piece. If $|C|=1$, we let $f$ be the composition of a sufficient number of two-piece left-rotations (\cref{obs:rot}) to place the single element of $C$ into the most significant position; in this case, $f^{-1}$ is just the composition of a complementary number of left-rotations, modulo~$k$. Otherwise, We construct the function $f$ as a composition of piecewise linear bijections as follows:
\begin{itemize}
\item Remove a single element from $C$, producing the smaller set $C'$. By induction, perform a composition $f'$ of piecewise linear bijections that places $C'$ into the most significant positions of the resulting permuted bit sequence. Let $i$ be the position in which these bijections leave the remaining element that was removed from $C$ to produce $C'$.
\item Perform $k-1-i$ left circular shifts using two-piece piecewise linear bijections, as described in \cref{obs:rot}. As a result, in the value resulting from these shifts, the bit that started out in position $i$ will be in the most significant position. The subinterval $[0,2^{k-1})$ will contain the inputs for which this most significant bit is zero, and the subinterval $[2^{k-1},2^k)$ will contain the inputs for which it is one. The remaining bits of $C$ will form a contiguous block elsewhere in the bit sequence; let $j$ be the number of positions separating the most significant bit from this contiguous block.
\item Perform $j$ left circular shifts of the low-order $k-1$ bits. Each of these circular shifts can be performed as a four-piece piecewise linear bijection, obtained by applying \cref{obs:rot} separately to the two subintervals $[0,2^{k-1})$ and $[2^{k-1},2^k)$.
\end{itemize}
To construct the inverse $f^{-1}$ of the function $f$ that we constructed above, we simply reverse these steps:
\begin{itemize}
\item Perform $k-1-j$ left circular shifts of the low-order $k-1$ bits of the given value.
\item Perform another $i+1$ left circular shifts to invert the effect of the $k-1-i$ left circular shifts included in $f$. 
\item Construct and perform the inverse of $f'$, by induction.
\end{itemize}
The $|C|$ levels of induction each add $O(k)$ circular shifts with two or four pieces each, from which the bounds on numbers of composed functions and the numbers of pieces follow.
\end{proof}

\subsection{Completeness}

We are now ready to prove that the piecewise linear bijection problem is $\ib$-complete. The main idea of the proof is to use \cref{lem:plb-permute} to simulate individual gates of a reversible logic circuit, and \cref{lem:compose-plb} to combine the resulting sequences of piecewise linear bijections into a single piecewise linear bijection for which it is hard to find iterated values.

\begin{theorem}
\label{thm:plb-complete}
The piecewise linear bijection problem is $\ib$-complete.
\end{theorem}

\begin{proof}
Because it is a problem of finding iterated values of a bijection that can be computed in polynomial time (\cref{obs:plb-poly}), it is clearly in $\ib$. We prove that it is complete by finding a many-one reduction from the problem of finding iterated values of reversible logic circuits, already proven $\ib$-complete (\cref{obs:circuit-completeness}). To do so, we suppose that we are given a reversible logic circuit operating on $k$-bit inputs and outputs, and we represent the values on sets of $k$ wires of the circuit (progressing from the $k$ input wires, through the various gates of the circuit, to the $k$ output wires, as numbers in the range $[0,2^k)$, interpreted as $k$-bit binary numbers. We will find sequences of  piecewise linear bijections that, under this interpretation, perform the operation of each gate of the given circuit, and compose these sequences into a single bijection using \cref{lem:compose-plb}.

In more detail, for each gate of the circuit, in a topological ordering of the circuit ordered from its inputs to its outputs, we apply the following sequence of piecewise linear bijections:
\begin{itemize}
\item Let $C$ be the set of bit positions representing the inputs and outputs of the gate. Apply \cref{lem:plb-permute} to find a sequence of piecewise linear bijections whose composition $f_C$ permutes all of the bits so that the positions of $C$ are permuted into the $|C|$ most significant positions.
\item The resulting values can be partitioned into $2^{|C|}$ subintervals of the form $[i2^{k-|C|},(i+1)2^{k-|C|})$ (for $i=0,1,\dots 2^{|C|}-1$, within which the bits in the $|C|$ most significant bit positions have constant values. Apply a single piecewise linear bijection that permutes these $2^{|C|}$ subintervals into an order representing the output values of the given gate for these input values.
\item Apply the sequence of piecewise linear bijections whose composition represents the inverse function $f^{-1}_C$ to the function $f_C$, as described by \cref{lem:plb-permute}.
\end{itemize}
If the circuit has $g$ gates, each operating on $O(1)$ bits, then the composition of the sequences described above for each gate gives us an overall sequence of $s=O(gk)$ piecewise linear transformations, with $O(gk)$ pieces, that implements the same function as the given circuit.

By \cref{lem:compose-plb} we can find an equivalent single piecewise linear transformation $T$, with the same number of pieces, operating on $O(gk^2)$-bit values, such that $T^{(s)}(x)$, for $x\in[0,2^k)$, performs a single iteration of the given circuit. Therefore, for any $n$, we can apply the circuit $n$ times to input value $x$, by solving the piecewise linear bijection problem of computing $T^{(ns)}(x)$. The construction of $T$ and $ns$ from the circuit are polynomial-time transformations, so this gives a many-one reduction from finding iterated values of reversible logic circuits to the piecewise linear bijection problem.
\end{proof}

\subsection{Integer interval exchange transformations}

We conclude this section with a special case of the piecewise linear bijection problem that has a non-obvious polynomial time algorithm. We define the \emph{iterated integer interval exchange transformation problem} to be the special case of the piecewise linear bijection problem in which each of the linear bijections of the given bijection has multiplier 1. In these piecewise linear bijections, each piece is just a translation, and the whole bijection is an interval exchange transformation. After an earlier version of this paper listed the complexity of the iterated integer interval exchange transformation problem as an open problem, Mark Bell provided in a personal communication the observation that it has a polynomial time solution, obtained by reinterpreting it as a problem on normal curves in triangulated surfaces and plugging in known results from computational topology. In this section we provide a more detailed expansion of Bell's observation.
To describe this solution, we need the following topological definitions, for which refer to \cref{fig:normal-interval-exchange} and \cite{EriNay-DCG-13}.

\begin{figure}[t]
\centering\includegraphics[width=0.6\textwidth]{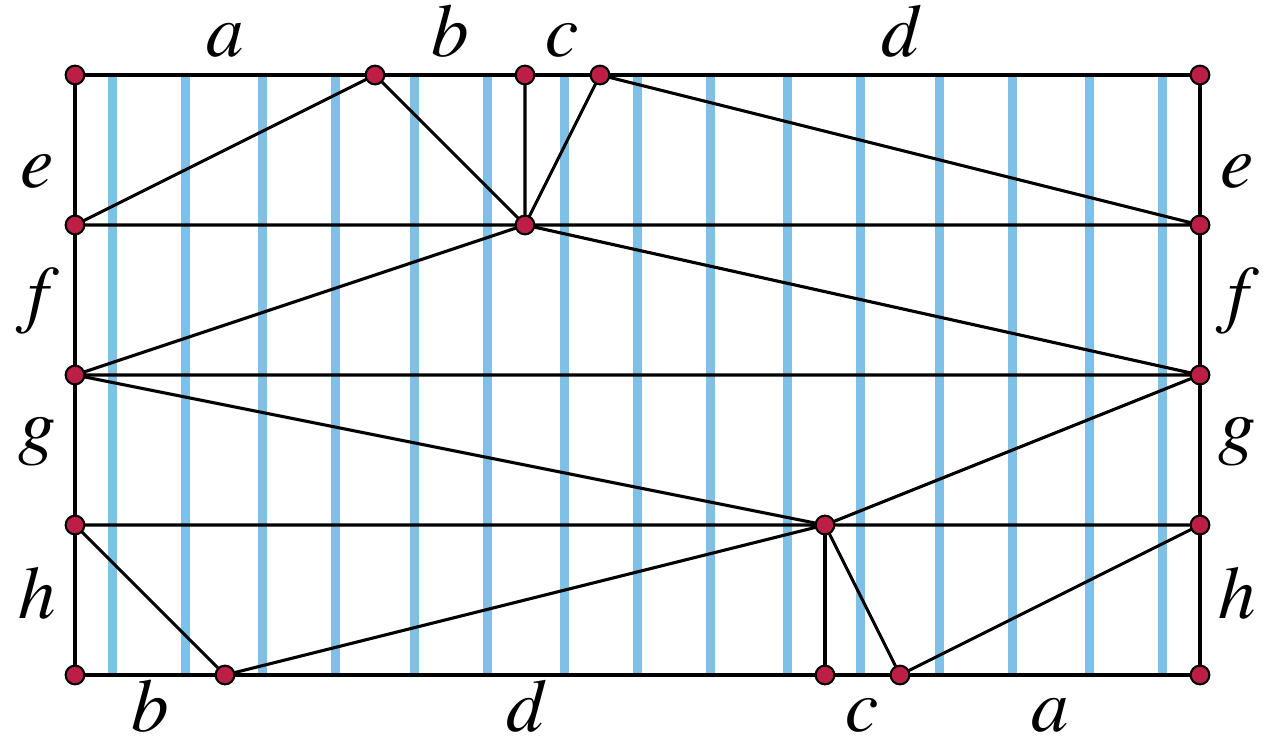}
\caption{A normal curve (light blue) on a triangulated double torus (black triangles and red vertices, glued from top to bottom and from left side to right side with the pairing indicated by the letters). Traversing the normal curve upwards from its central horizontal line, through the glued edges from top to bottom, and continuing upwards back to the same central line, permutes the branches of the curve according to the integer interval exchange transformation that maps $[0,3]\mapsto [11,14]$, $[4,5]\mapsto [0,1]$, $6\mapsto 10$, and $[7,14]\mapsto [2,9]$.}
\label{fig:normal-interval-exchange}
\end{figure}

\begin{definition}
A \emph{triangulated oriented 2-manifold} (``triangulated manifold'', for short) is a topological space obtained from a system of triangles, each having a specified clockwise orientation of its edges, and a matching of pairs of edges of those triangles to be glued together consistently with those orientations. (The vertices of the manifold are equivalence classes of vertices of triangles after this gluing, but are unimportant for what follows.) A \emph{normal curve}  is a one-dimensional subspace of a triangulated manifold, topologically equivalent within each triangle to a collection of disjoint line segments that extend from edge to edge within the triangle, avoiding its vertices. At the point where one of these segments meets an edge of a triangle, it is required to continue from the same point into another segment in the other triangle glued to the same edge. In particular, because a line segment cannot cross the same edge of a triangle twice in succession, a normal curve is also forbidden from having segments that do this, although it can return to an edge after crossings with other edges. A normal curve can have multiple connected components; a single component is called an \emph{arc}.
\end{definition}

\begin{definition}
The \emph{normal coordinates} of a normal curve on a triangulated surface are a labeling of each edge $e$ of the triangulation by a non-negative integer $N_e$, the number of points of intersection between the curve and edge $e$ (\cref{fig:normal-coords}).
\end{definition}

The following is standard:

\begin{observation}
The normal coordinates of any normal curve of any triangulated surface obey the triangle inequality in each triangle of the surface, and sum to an even number in each triangle of the surface. Any system of non-negative integer edge labels obeying these constraints defines a normal curve, which is unique up to homeomorphisms of the surface that map each vertex and edge of the triangulation to itself.
\end{observation}

\begin{figure}[t]
\centering\includegraphics[width=0.3\textwidth]{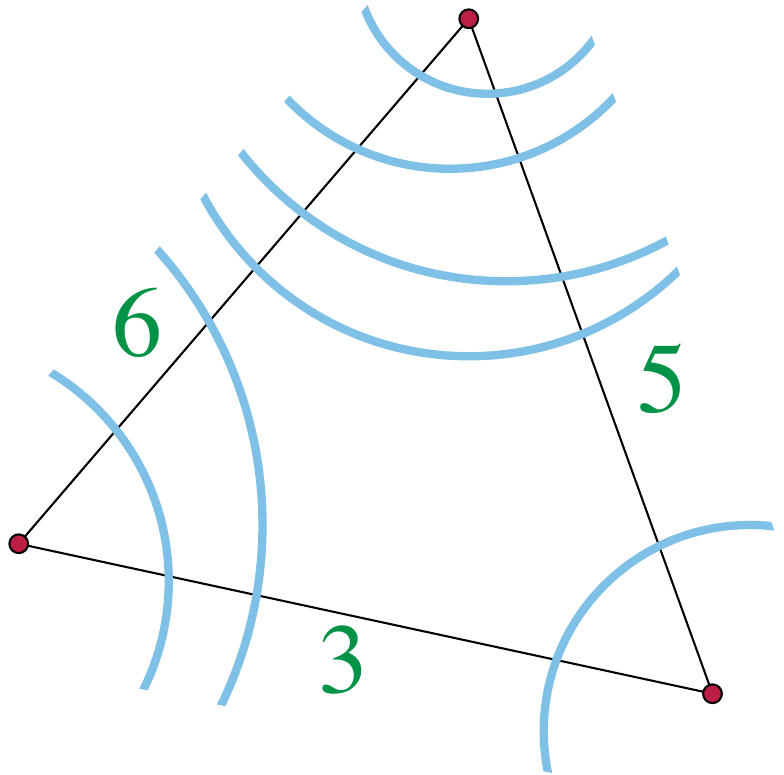}
\caption{A triangle in a triangulated surface (black), part of a normal curve (blue), and the normal coordinates of the three triangle edges (green)}
\label{fig:normal-coords}
\end{figure}

\begin{proof}
Consider any normal curve $C$ of any triangulated surface, and a triangle $\Delta$ of the triangulation with edges $x$, $y$, and $z$. Within $\Delta$, $C$ must consist of some number $a\ge 0$ of segments crossing from $x$ to $y$, some number $b\ge 0$ of segments crossing from $y$ to $z$, and some number $c$ of segments crossing from $x$ to $z$. Then its normal coordinates on these three edges are $N_x=a+c$, $N_y=a+b$, and $N_z=b+c$. Their sum is $2a+2b+2c$, an even number. They obey the triangle inequality because $N_x+N_y=2a+b+c\ge b+c=N_z$.

Conversely, suppose we are given any system of normal coordinates that obey the triangle inequality and sum to an even number in each triangle.
Then because the sum $N_x+N_y+N_z$ is even, by assumption, $N_x+N_y-N_z$ is even, as it differs from $N_x+N_y+N_z$ by the even number $2N_z$. Additionally, $N_x+N_y-N_z$ is non-negative, by the triangle inequality. It follows that if we set $a=(N_x+N_y-N_z)/2$, and symmetrically $b=(N_x-N_y+N_z)/2$ and $c=(-N_x+N_y+N_z)/2$, then $a$, $b$, and $c$ are non-negative integers. We can construct a normal curve having this system of normal coordinates by placing $N_e$ crossing points on each edge $e$. Then, within each triangle $xyz$ with $a$, $b$, and $c$ calculated as above, we draw $a$ line segments connecting the crossing points on edges $x$ and $y$ that are the $a$ nearest crossings to the shared vertex of $x$ and $y$. Symmetrically, we draw $b$ line segments connecting the crossing points on edges $y$ and $z$ that are the $b$ nearest crossings to the shared vertex of $y$ and $z$, and we draw $c$ line segments connecting the crossing points on edges $x$ and $z$ that are the $c$ nearest crossings to the shared vertex of $x$ and $z$. The resulting system of line segments within each triangle link up to form a normal curve, whose normal coordinates as calculated above are exactly the numbers we are given.

Any two normal curves with the same coordinates necessarily have the same number of crossing points on each edge of the triangulation (given by the normal coordinates) and the same pattern of segments of the curve within each triangle (as described above).
They can be mapped to each other by a homeomorphism of the surface
that fixes the vertices of the triangulation, maps each edge to itself in a way that that takes the crossing points of one normal curve to the crossing points of the other normal curve, and then maps the interior of each triangle to itself in a way that deforms one system of segments from one normal curve into the corresponding system of segments from the other normal curve.
\end{proof}

In order to analyze the computational complexity of algorithms on normal curves we also need the following.

\begin{observation}
\label{obs:normal-bit-complexity}
The number of bits needed to specify a triangulated surface with $\tau$ triangles, and a normal curve on that surface with $\sigma$ segments,
is $O(\tau\log(\tau + \sigma))$.
\end{observation}

\begin{proof}
The surface can be specified by numbering and orienting the triangles and, for each triangle, specifying its three neighboring triangles. This specification uses $3\lceil\log_2\tau\rceil$ bits per triangle.
Additionally, each normal coordinate is at most $\sigma$ and specifying it takes at most $1+\log_2\sigma$ bits for each of the $3\tau/2$ edges.
\end{proof}

We need to specify, not only a curve on a surface, but a crossing point of the curve with an edge of the triangulated surface. To do so, it is helpful to introduce index numbers for these crossing points. We use two different forms of indexing, \emph{edge coordinates} and \emph{arc coordinates}.

\begin{definition}
Choose an arbitrary orientation for each edge of a triangulated surface with a specified normal curve. Then, for this orientation,
the edge coordinate of a point where a normal curve crosses an edge is just its position among the crossings on that edge, after choosing an arbitrary orientation for that edge.

Similarly, choose an arbitrary orientation for each arc of the specified normal curve, and designate one of the crossings points of each arc  (chosen arbitrarily) as its starting point. Then, for this data, the arc coordinate of a point where a normal arc crosses an edge is its position among all of the crossings along the normal arc, in the order they are reached by following that arc from its starting point in the direction of the specified orientation.
\end{definition}

Erickson and Nayyeri~\cite{EriNay-DCG-13} provide several useful algorithms for manipulating normal curves, normal coordinates, and the edge coordinates and arc coordinates of their crossings. In particular, they show that the following computations can all be done in time polynomial in the bit complexity of the normal curve, as given by \cref{obs:normal-bit-complexity}:
\begin{itemize}
\item Given a normal curve and the edge coordinate of a crossing point $p$ in this curve, find the normal coordinates that describe the arc of the normal curve containing $p$ \cite[Theorem 6.2]{EriNay-DCG-13}. The algorithm constructs a \emph{street complex} describing this arc, from which it is possible to convert edge coordinates in the given curve into edge coordinates in the arc and vice versa.
\item Given a normal curve consisting of a single arc, the edge coordinate of a crossing point $p$, and a choice of a starting point on that arc,
find the arc coordinate of $p$ \cite[Theorem 6.3]{EriNay-DCG-13}.
\item Given a normal curve consisting of a single arc, the arc coordinate of a crossing point $p$, and a choice of a starting point on that arc,
find the edge coordinate of $p$ \cite[Theorem 6.4]{EriNay-DCG-13}.
\end{itemize}
More precisely, the time for each of these operations is quadratic in the number of triangles in the triangulation and logarithmic in the total number of crossings of the normal curve.

\begin{theorem}[Bell]
\label{thm:integer-exchange}
The iterated integer interval exchange transformation problem, for $n$ iterations of a transformation on $k$ intervals over the integers in the range $[0,N-1]$, can be solved in time polynomial in $k$, $\log N$, and $\log n$.
\end{theorem}

\begin{proof}
Given an integer interval exchange transformation $f$, a number of iterations $n$, and a starting value $i$, perform the following steps to compute $f^{(n)}(i)$:
\begin{itemize}
\item Construct a surface in the form depicted in \cref{fig:normal-interval-exchange}: a triangulated rectangle, divided into some number of horizontal stripes, with the input intervals to the transformation subdividing the top of the rectangle and the outputs subdividing the bottom. It is convenient to assign Cartesian coordinates to this rectangle in such a way that the vertical blue lines have integer $x$-coordinates in the range $[0,N-1]$. To do so, place the left edge of the rectangle on the line $x=-1/2$ of the Cartesian plane and its right edge on the line $x=N-1/2$.  Triangulate it as shown in the figure, so that each vertical line through the rectangle crosses exactly one diagonal edge per slice, and so that the central horizontal line across the rectangle forms an unsubdivided edge. Each step up or down from this central edge allows the number of horizontal subdivisions to double, so it suffices to take $s=\lceil\log_2 k\rceil$ steps above and below the central edge, producing a triangulation with  $O(k)$ triangles. Glue the left and right sides of the rectangle together, and glue top and bottom according to the labeling from the transformation.
\item Compute the normal coordinates of the normal curve formed by the vertical integer lines in the rectangle. The normal coordinate of any edge is how far apart horizontally its endpoints are within the rectangle.
\item Given the integer $i$ whose iterates we want to compute, let $p_i$ be the crossing point on the central horizontal edge with $x$-coordinate equal to~$i$. Each iteration of the exchange transformation can be obtained by advancing $2s$ units upward along the curve, starting from $p_i$, and wrapping around halfway through from the top of the rectangle to the bottom.
\item Find the normal arc containing $p_i$ and its normal coordinates (the number of times this arc crosses each edge of the triangulation), using the algorithm of \cite[Theorem 6.2]{EriNay-DCG-13}. The length $\ell$ of this arc (measured as a number of crossings) is just the sum of these normal coordinates. The whole normal curve has $O(N\log k)$ crossings so this bound applies \emph{a fortiori} to this arc.
\item Convert the known edge coordinate~$i$ of $p_i$ to an arc coordinate using the algorithm of \cite[Theorem 6.3]{EriNay-DCG-13}. Add $2ns$ (modulo $\ell$) to this arc coordinate, and convert the resulting arc coordinate back into an edge coordinate in its arc using the algorithm of \cite[Theorem 6.4]{EriNay-DCG-13}. Reversing the reduction from the normal curve to a normal arc, convert this edge coordinate in the arc back into an edge coordinate in the whole curve. This edge coordinate is the number $f^{(n)}(i)$ that we want to compute.
\end{itemize}

The topological subroutines used by this algorithm all take time quadratic in the size of the triangulation and logarithmic in the number of crossings of the normal curve, so for the triangulation and normal curve constructed above this time bound is $O(k^2\log N)$.
\end{proof}

\section{Conclusions}

We have studied the complexity of computing the iterated values of polynomial-time bijections, shown the equivalence of several different complexity classes defined in this way to each other and to $\ib$, and found $\ib$-complete problems involving circuits, cellular automata, graphs, and piecewise linear transformations. Our cellular automaton completeness result, in particular, involves a novel dimension-reducing simulation from two-dimensional to one-dimensional reversible cellular automata. Our results may be of interest in expanding the types of problems that are complete for this class.

We have also described a class of problems of this type, the iterated integer interval exchange transformation problems, that have a non-obvious polynomial-time algorithm. In subsequent work, we have applied this algorithm as a subroutine to solve certain problems of ray-tracing in two-dimensional systems of mirrors~\cite{Epp-CCCG-22}.

\printbibliography
\end{document}